\documentclass[fleqn,11pt]{article}

\usepackage[english]{babel}

\usepackage[
    a4paper,
    top=30mm,
    bottom=30mm,
    left=30mm,
    right=30mm,
    footskip=12mm
]{geometry}

\usepackage[inline]{enumitem}
\usepackage{arydshln}
\usepackage{amsmath}
\usepackage{amssymb}
\usepackage{amsfonts}
\usepackage[bb=boondox]{mathalfa}

\usepackage{bm}
\newcommand{\symbfit}[1]{\bm{#1}}

\usepackage{amsthm}
\usepackage{mathtools}

\theoremstyle{plain}
\newtheorem{theorem}{Theorem}[section]

\theoremstyle{definition}

\theoremstyle{remark}
\newtheorem*{remark}{Remark}

\usepackage{graphicx}
\usepackage[hang,small,bf]{caption}
\usepackage[subrefformat=parens]{subcaption}
\usepackage{rotating}
\usepackage[authoryear,round]{natbib}

\usepackage[%
 setpagesize=false,%
 bookmarks=true,%
 bookmarksdepth=tocdepth,%
 bookmarksnumbered=true,%
 colorlinks=true,%
 linkcolor=red,%
 citecolor=blue,%
 urlcolor=magenta,%
 pdftitle={},%
 pdfsubject={},%
 pdfauthor={},%
 pdfkeywords={}%
]{hyperref}

\begin{document}

\title{ 
Extended State-dependent Hawkes Process for Limit Order Books: 
Mathematical Foundation and the Reproduction of Volatility Signature Plots
\footnote{This work was supported by JSPS KAKENHI Grant Number JP20K14366 
and CREST, JST. 
The ``Nikkei NEEDS Tick Data'' used in this study was 
purchased through the aforementioned KAKENHI grant.}
}

\author{
Akitoshi Kimura 
\thanks{Sophia University
Department of Information and Communication Sciences
Faculty of Science and Technology
Mail: a\_kimura@sophia.ac.jp
Japan Science and Technology Agency, CREST}
}
\date{\today}

\maketitle

\begin{abstract}
  This paper proposes an Extended State-Dependent Hawkes Process (ExsdHawkes) 
  to model the intricate dynamics of Limit Order Books (LOBs). 
  Our theoretical contribution lies in relaxing traditional constraints 
  by allowing for state disappearances---a phenomenon frequently observed 
  in high-frequency trading. We mathematically prove, 
  using Karush--Kuhn--Tucker (KKT) conditions, 
  that the maximum likelihood estimation remains separable, 
  justifying an efficient two-step procedure. 
  
  In the empirical section, we apply our model 
  to three months of high-frequency tick data of Mitsubishi UFJ Financial Group (8306). 
  We demonstrate that ExsdHawkes successfully replicates 
  the characteristic upward slope of the volatility signature plot 
  by capturing the ``local super-criticality'' triggered during disequilibrium states. 
  Crucially, we clarify that the transition out of equilibrium is 
  deterministically triggered by Aggressive Market Orders (AMS/AMB), 
  while Marketable Limit Orders (MLO) function 
  as a critical liquidity-depletion catalyst within the expanded spread. 
  Comparative analysis reveals that models lacking physical constraints 
  (e.g., standard SD-Hawkes) 
  suffer from explosive spectral radii and fail to maintain simulation stability.
  Our findings suggest that physical consistency is not merely a mathematical nicety, 
  but a prerequisite for accurately modeling macro-level volatility. 
  By enforcing the physical geometry to `pause' 
  the residual accumulation during inadmissible periods, 
  ExsdHawkes maintains statistical integrity 
  where unconstrained models succumb to structural bias and simulation instability.
\end{abstract} 

\section{Introduction}\label{sec:Introduction}
The dynamics of Limit Order Books (LOBs) have been a focal point of 
market microstructure research, particularly in understanding how 
discrete order arrivals translate into continuous price volatility. 
One of the most intriguing phenomena in this domain is 
the upward slope of the volatility signature plot---the observation that 
realized volatility increases significantly as the sampling frequency rises. 
While various stochastic models, including Hawkes processes, 
have attempted to explain this behavior, 
capturing the precise feedback loops between micro-level order flow 
and the macro-level volatility structure remains a significant challenge.

The primary difficulty lies in the fact that 
high-frequency volatility is not uniformly distributed over time. 
Empirical evidence suggests that volatility spikes are often concentrated in 
short-lived periods of market disequilibrium. 
In this paper, we propose that the origin of the upward slope in signature plots 
can be traced back to a specific three-stage dynamic process:
\begin{enumerate}
  \item Trigger: 
  The arrival of Aggressive Market Orders (AMS/AMB) 
  instantaneously consumes the prevailing best quotes, 
  deterministically forcing the LOB out of its equilibrium state ($x = 1$) 
  into a disequilibrium state ($x = 2+$).
  \item Catalyst and Transition: 
  Irrespective of whether the pre-arrival spread is 
  at the minimum tick ($x = 1$) or already widened, 
  Marketable Limit Orders (MLO) function as a severe liquidity-depletion catalyst 
  by rewriting price boundaries. 
  Due to the temporary removal or expansion of physical constraints, 
  order intensities exhibit explosive clustering governed 
  by a super-critical localized spectral radius ($\rho(2+) > 1.0$).
  \item Local Super-criticality: 
  This phase of ``local super-criticality'' generates intense price fluctuations, 
  which manifest as the characteristic upward slope when sampled at high frequencies.
\end{enumerate}

Conventional state-dependent models, 
such as \citet{morariu2022state}, 
have laid the foundation for capturing these dynamics. 
However, they often impose rigid row-sum constraints on transition probabilities 
to ensure mathematical tractability, 
which can lead to physical inconsistencies---such as 
assigning non-zero probability to price-improving orders 
when the spread is already at its minimum tick size.

Our contribution is twofold. 
First, we propose an Extended State-dependent Hawkes Process (ExsdHawkes) 
that rigorously incorporates the physical constraints of LOBs. 
By applying Karush--Kuhn--Tucker (KKT) conditions, 
we show that transition probabilities can be estimated separately 
from Hawkes parameters even when some transitions are physically prohibited. 
Second, we apply this model to three months of MUFG tick data and demonstrate 
that incorporating physical consistency provides a robust framework 
to remain stable during local super-critical phases and replicate 
the macro-level features of the volatility signature plot.

\section{Literature Review}\label{sec:Literature}

\subsection{Endogeneity and the Micro-structure of Volatility}\label{subsec:Endogeneity}
The endogenous nature of financial markets---where the arrival of an order is both a 
reaction to past events and a catalyst for future activity---has been a 
cornerstone of market microstructure theory. 
\citet{bacry2015hawkes} 
established the Hawkes process as the standard mathematical 
tool for capturing this ``reflexivity'', 
demonstrating that self-exciting point processes can explain the stylized facts of 
order flow clustering more effectively than traditional Poisson paradigms.

However, the connection between these micro-level excitations and 
macro-level price structures remains a subject of intense debate. 
Early studies focused primarily on permanent price impact and 
the long-range dependence of order signs. 
More recently, the focus has shifted toward how the 
high-frequency feedback loop contributes to realized volatility. 
A critical insight provided by 
\citet{filimonov2012quantifying} 
is that the ``reflexivity index'' (or branching ratio) of the market has 
increased over time due to the rise of algorithmic trading, 
pushing the system closer to criticality.

In this context, the volatility signature plot---which depicts realized volatility 
as a function of sampling frequency---serves as a crucial diagnostic tool. 
While standard Hawkes models can theoretically generate an 
upward slope in signature plots through their self-exciting mechanism, 
they often fail to capture the state-dependent ``bursts'' of 
activity that characterize real LOB dynamics. 
Our work extends this lineage by proposing that the macro-level upward slope 
is not just a result of general self-excitation, 
but is specifically driven by transient phases of local super-criticality 
triggered by liquidity-consuming orders in disequilibrium states. 
By situating our model within this endogeneity debate, 
we provide a structural explanation for why the localized spectral radius surges precisely 
when the LOB physical geometry is most stressed.

\subsection{Evolution of State-Dependency and the Physical Constraint Dilemma}\label{subsec:Evolution}
The integration of LOB state information into the Hawkes framework 
has evolved from simple baseline adjustments to complex, kernel-modulated systems. 
Early influential work by \citet{large2007measuring} 
utilized Hawkes processes to measure the resiliency of the LOB, 
albeit without a formal state-dependent kernel structure. 
Subsequently, models began to incorporate more granular state variables, 
such as the order book imbalance and the prevailing spread, 
to capture the varying ``market climate''. 

A landmark advancement was achieved by \citet{morariu2022state}, 
who established a rigorous foundation for state-dependent Hawkes processes 
where both the baseline intensities and the kernels are modulated by the LOB state. 
However, to ensure mathematical tractability---specifically to maintain the 
existence of a stationary distribution---conventional models have typically 
imposed a rigid row-sum constraint on transition probabilities, 
requiring that $\Phi_{e, x} = \sum_{x'} \phi_e(x, x')$ always equals unity. 

While this constraint simplifies stability analysis, 
it creates a fundamental  ``physical constraint dilemma'' in liquid markets. 
As highlighted by \citet{cont2011statistical}, 
the LOB is governed by rigid geometric rules, 
such as the minimum tick size. 
At the minimum spread ($x = 1$), 
a price-improving order (ALS or ALB) is physically impossible. 
Conventional models that force $\Phi_{e, x} = 1$ are thus compelled to 
assign non-zero probability to these impossible transitions.

This structural misalignment results in intensity leakage: 
the model predicts event arrivals in periods where the LOB's 
physical geometry strictly forbids them. 
Such leakage introduces systematic bias into the Hawkes kernels, 
as the estimation process attempts to ``compensate'' for the absence of 
observed events in inadmissible states by skewing the endogenous parameters. 
Our ExsdHawkes framework resolves this by allowing for 
``state disappearances''($\Phi_{e, x} = 0$), 
treating the LOB geometry as a first-order principle that precedes the stochastic dynamics.

\subsection{Refining Order Taxonomy: The Role of MLO}\label{subsec:Refining}
In most seminal limit order book (LOB) models, 
such as those by \citet{abergel2016limit} 
and \citet{lu2018high}, 
orders are typically categorized into a parsimonious three-tier taxonomy: 
limit orders, market orders, and cancellations. 
While this simplification ensures mathematical elegance, 
it fails to account for the unique informational and structural impact of 
Marketable Limit Orders (MLO).

Unlike a pure market order that only consumes liquidity at the best quote, 
an MLO is a hybrid instrument. 
It acts as a market order by triggering immediate execution, 
but simultaneously functions as a limit order that resets 
the price boundary on its own side. 
From a micro-structural perspective, 
an MLO is the primary instrument of liquidity depletion and price discovery. 
Specifically, as empirically demonstrated by the transition probabilities 
in Section~\ref{sec:Empirical}, 
Aggressive Market Orders (AMS/AMB) serve as the absolute boundary-breaking trigger, 
whereas MLOs function as the structural catalyst that deepens the volatility 
and blocks immediate quote recovery.

This ``leap'' in the spread is a critical event that is 
often obscured when MLOs are lumped together with standard market orders. 
In many high-frequency environments, 
the transition to a wider spread temporarily removes the ``buffer'' of liquidity, 
allowing for the transient phases of local super-criticality we observe. 
By isolating MLOs as a distinct event type within our state-dependent framework, 
we can model the LOB not as a static queue, 
but as a dynamic system characterized by rapid alternations 
between liquidity-rich stability and liquidity-sparse volatility. 
This refined taxonomy allows us to pinpoint the exact micro-level catalyst responsible 
for the macro-level upward slope of the volatility signature plot.

\subsection{Event Taxonomy and Microstructure Classification}
\label{sec:event_taxonomy}

To bridge our theoretical framework with empirical high-frequency dynamics, 
we formalize the order flow state space $\mathcal{E}$ 
under a comprehensive $d_e = 14$ dimensional taxonomy. 
Following the event-driven classification principles of 
\citet{lu2018high}, 
we separate the incoming order book ticks into seven execution mechanisms 
across both sides of the book (Seller/Buyer):
\begin{itemize}
  \item \textbf{Non-aggressive Market Orders (NMS / NMB):} 
  Standard market transactions that consume liquidity strictly 
  within the volume boundaries of the prevailing best quotes without altering the mid-price.
  \item \textbf{Non-aggressive Limit Orders (NLS / NLB):} 
  Standard passive limit orders placed at or deeper than the current best quotes, 
  widening queue depth without narrowing the spread.
  \item \textbf{Non-aggressive Cancellations (NCS / NCB):} 
  Standard cancellations of passive limit orders 
  that do not alter the prevailing spread. 
  This category comprehensively includes cancellations executed deeper inside the order queues, 
  as well as partial or full order cancellations at the best quotes 
  that leave a non-zero residual volume, thereby causing no immediate price shift. 
  \item \textbf{Aggressive Market Orders (AMS / AMB):} 
  High-volume market orders that exhaust the entire liquidity available at the best quote, 
  acting as the deterministic trigger that widens the spread.
  \item \textbf{Aggressive Limit Orders (ALS / ALB):} 
  Price-improving limit orders placed directly inside the spread, 
  causing an immediate narrowing of the price gap.
  \item \textbf{Aggressive Cancellations (ACS / ACB):} 
  Volatility-inducing cancellations that target the strict remaining total volume at the best quote, 
  completely exhausting the queue and instantaneously forcing a spread-widening transition.
  \item \textbf{Marketable Limit Orders (MLO / MLB):} 
  Hybrid instruments crossing the spread for immediate execution 
  while simultaneously registering any residual mass as a new quote boundary.
\end{itemize}
This precise parameter categorization allows us 
to systematically audit the specific multi-layered excitation paths responsible 
for macro-level price diffusion without losing mathematical elegance.

\section{Methodology}\label{sec:Methodology}

\subsection{Definition of ExsdHawkes}
\label{subsec:Definition}

The Extended State-Dependent Hawkes Process (ExsdHawkes) incorporates 
the physical geometry of the LOB 
by relaxing the traditional row-sum restrictions on transition structures. 
Let $\{T_{n}\}$ be the discrete event arrival times, 
$\{E_{n}\}$ the event types belonging to a finite event space $\mathcal{E}$ 
with dimension $d_e := |\mathcal{E}|$, 
and $\{X_{n}\}$ the LOB states defined on a finite state space $\mathcal{X}$ 
with dimension $d_x := |\mathcal{X}|$. 
As explicitly classified under the microstructure taxonomy 
in Section~\ref{sec:event_taxonomy}, 
our empirical asset application constructs a $d_e = 14$ dimensional interaction framework.

For a given parameterized baseline vector 
$\symbfit{\nu} \in \mathbb{R}^{d_e}$ 
and the endogenous Hawkes kernel matrix 
$k_{e'e}^{\symbfit{\theta}}(t, x)$ 
governed by the compact memory decay subspace 
$\symbfit{\theta} \in \mathcal{B}_{H}$, 
the state-dependent intensity 
$\tilde{\lambda}_{ex}(t)$ 
for a given pair 
$(e, x) \in \mathcal{E} \times \mathcal{X}$ 
is formally defined as:
\begin{equation}
  \tilde{\lambda}_{ex}(t) 
  = \phi_{e}(X(t), x) \left( \nu_{e} 
  + \sum_{e' \in \mathcal{E}} \sum_{x' \in \mathcal{X}} 
  \int_{[0,t)} k_{e'e}^{\symbfit{\theta}}(t - s, x') d\tilde{N}_{e'x'}(s) \right), 
  \label{eq:exsd_intensity_def}
\end{equation}
where $\phi_{e}(x, x') \in [0, 1]$ represents the state transition probability weight. 
Critically, we relax the unconstrained row-sum structure 
by enforcing the binary condition 
$\Phi_{e,x} := \sum_{x' \in \mathcal{X}} \phi_{e}(x, x') \in \{0, 1\}$, 
allowing for complete state disappearances $(\Phi_{e,x} = 0)$ 
where specific execution paths are geometrically prohibited 
by the rigid guidelines of the LOB tick size.

\subsection{Global Separability of the Likelihood Parameter Space}\label{subsec:global_separability}
A central challenge in state-dependent point process modeling is ensuring that 
the empirical estimation of transition probabilities does not introduce 
structural bias into the endogenous Hawkes parameters. 
To achieve rigorous parameter orthogonality, 
we define the full mother parameter space $\mathcal{B}_{total}$ 
as the strict Cartesian product of two mutually independent sub-spaces: 
$\mathcal{B}_{total} = \mathcal{B}_{\phi} \times \mathcal{B}_{H}$, 
where $\symbfit{\phi} \in \mathcal{B}_{\phi}$ governs the transition networks, 
and $(\symbfit{\nu}, \symbfit{\theta}) \in \mathcal{B}_{H}$ 
entirely controls the endogenous baseline and memory decay infrastructure. 

To facilitate this, we decompose the state-dependent intensity $\tilde{\lambda}_{ex}(t)$ 
from Eq. (\ref{eq:exsd_intensity_def}) 
into the admissibility-weighted intensity $\lambda^{\dagger}_e(t)$:
\begin{equation}
\lambda^{\dagger}_e(t \mid \symbfit{\nu}, \symbfit{\theta}) 
= \sum_{x \in \mathcal{X}} \tilde{\lambda}_{ex}(t) 
= \Phi_{e,X(t)} \left( \nu_e 
+ \sum_{e', x'} \int_{[0,t)} k^{\symbfit{\theta}}_{e'e}(t - s, x') d\tilde{N}_{e'x'}(s) \right) 
\end{equation}
where $\Phi_{e,X(t)} = \sum_{x' \in \mathcal{X}} \phi_e(X(t), x') \in \{0, 1\}$ 
acts as the deterministic indicator of whether event $e$ is physically admissible 
in the current state.

\begin{theorem}[Separability of Estimation] \label{thm:3.1}
  The full log-likelihood 
  $\ln \mathcal{L}(\symbfit{\phi}, \symbfit{\nu}, \symbfit{\theta})$
  can be decomposed into two independent terms: 
  \begin{equation*}
    \ln \mathcal{L}(\symbfit{\phi}, \symbfit{\nu}, \symbfit{\theta})
    = \ln \mathcal{L}_{TP}(\symbfit{\phi}) 
    + \ln \mathcal{L}_H(\symbfit{\nu}, \symbfit{\theta}),  
  \end{equation*}
  where $\mathcal{L}_{TP}$ depends only on the transition counts and 
  $\mathcal{L}_H$ on the Hawkes parameters. 
\end{theorem}

\begin{proof}
Following the lift representation established 
in Section~\ref{subsec:Evolution} 
and Section~\ref{subsec:Definition}, 
a given realization $(t_n, e_n, x_n)_{n=1}^N$ 
over the operational horizon $[0, T]$ corresponds to the multivariate point process 
$\tilde{N} = (\tilde{N}_{ex})_{e\in\mathcal{E}, x\in\mathcal{X}}$ 
with the strictly coupled intensity $\tilde{\lambda}_{ex}(t)$ 
defined in Eq. (\ref{eq:exsd_intensity_def}), 
which is predictable with respect to the internal filtration 
$\mathcal{F}^{\tilde{N}}_t := \sigma(\tilde{N}(s) : 0 \le s \le t)$ 
generated by the historical event stream. 
By applying the density of the Janossy measure 
with respect to the Lebesgue measure on 
$\mathbb{R}^N$ \citep[p.~251]{daley2003introduction}, 
the rigorous full log-likelihood function 
$\ln \mathcal{L}(\symbfit{\phi}, \symbfit{\nu}, \symbfit{\theta})$ 
is formulated as:
\begin{equation}
\ln \mathcal{L}(\symbfit{\phi}, \symbfit{\nu}, \symbfit{\theta}) 
= \sum_{n=1}^N \ln \tilde{\lambda}_{E_n X_n}(T_n) - 
\int_0^T \sum_{e \in \mathcal{E}} \sum_{x \in \mathcal{X}} \tilde{\lambda}_{ex}(t) dt. 
\label{eq:global_likelihood_integral}
\end{equation}

We substitute the multi-layered intensity product representation 
$\tilde{\lambda}_{ex}(t) 
= \phi_e(X(t), x) \lambda^{\dagger}_{e}(t \mid \symbfit{\nu}, \symbfit{\theta})$ 
from Eq. (\ref{eq:exsd_intensity_def}) 
back into both the discrete event sum and the continuous compensator integral of 
Eq. (\ref{eq:global_likelihood_integral}). 
The first discrete term expands identically as follows:
\begin{align}
\sum_{n=1}^N \ln \tilde{\lambda}_{E_n X_n}(T_n) 
&= \sum_{n=1}^N \ln \left[ \phi_{E_n}(X_{n-1}, X_n) \left( \nu_{E_n} 
+ \sum_{e', x'} \int_{[0,T_n)} k^{\symbfit{\theta}}_{e'e}(T_n - s, x') d\tilde{N}_{e'x'}(s) 
\right) \right] \nonumber \\
&= \sum_{n=1}^N \ln \phi_{E_n}(X_{n-1}, X_n) 
+ \sum_{n=1}^N \ln \left( \nu_{E_n} + \sum_{e', x'} 
\int_{[0,T_n)} k^{\symbfit{\theta}}_{e'e}(T_n - s, x') d\tilde{N}_{e'x'}(s) \right). 
\label{eq:discrete_decomposition}
\end{align}

Concurrently, we evaluate the continuous compensator integral term 
over the joint system-wide space. 
By reversing the order of summation and applying the strict physical row-sum restriction 
$\sum_{x \in \mathcal{X}} \phi_e(X(t), x) = \Phi_{e,X(t)}$, 
the integral term simplifies under the strict assumption of LOB geometric boundaries:
\begin{align}
\int_0^T \sum_{e \in \mathcal{E}} \sum_{x \in \mathcal{X}} \tilde{\lambda}_{ex}(t) dt 
&= \int_0^T \sum_{e \in \mathcal{E}} \left[ \sum_{x \in \mathcal{X}} \phi_e(X(t), x) \right] 
  \left( \nu_e + \sum_{e', x'} 
  \int_{[0,t)} k^{\symbfit{\theta}}_{e'e}(t - s, x') d\tilde{N}_{e'x'}(s) 
  \right) dt \nonumber \\
&= \int_0^T \sum_{e \in \mathcal{E}} \Phi_{e,X(t)} \left( \nu_e + \sum_{e', x'} 
\int_{[0,t)} k^{\symbfit{\theta}}_{e'e}(t - s, x') d\tilde{N}_{e'x'}(s) \right) dt. 
\label{eq:continuous_decomposition}
\end{align}

Crucially, because the physical gate indicator $\Phi_{e,X(t)} \in \{0, 1\}$ is a deterministic, 
step-constant property governed exclusively by the rigid physical rules of the LOB tick size, 
it remains entirely pre-determined at time $t$ by the current state $X(t)$. 
Consequently, the continuous integral in Eq. (\ref{eq:continuous_decomposition}) 
depends strictly on the binary admissibility switches $\{0, 1\}$ 
and is completely invariant to the internal fractional probability weights $\phi_e(x, x')$. 

By regrouping the $\phi$-dependent terms 
from Eq. (\ref{eq:discrete_decomposition}) 
and the $(\symbfit{\nu}, \symbfit{\theta})$-dependent terms from both 
Eq. (\ref{eq:discrete_decomposition}) 
and Eq. (\ref{eq:continuous_decomposition}), 
the full log-likelihood function decomposes into two mutually orthogonal, 
structurally decoupled functional terms:
\begin{equation}
\ln \mathcal{L}(\symbfit{\phi}, \symbfit{\nu}, \symbfit{\theta}) 
= \ln \mathcal{L}_{TP}(\symbfit{\phi}) 
+ \ln \mathcal{L}_{H}(\symbfit{\nu}, \symbfit{\theta}),
\end{equation}
where the transition probability sub-likelihood $\ln \mathcal{L}_{TP}(\symbfit{\phi})$ and 
the Hawkes sub-likelihood $\ln \mathcal{L}_{H}(\symbfit{\nu}, \symbfit{\theta})$ 
are defined respectively as:
\begin{align}
\ln \mathcal{L}_{TP}(\symbfit{\phi}) 
&:= \sum_{n=1}^N \ln \phi_{E_n}(X_{n-1}, X_n), \\
\ln \mathcal{L}_{H}(\symbfit{\nu}, \symbfit{\theta}) 
&:= \sum_{n=1}^N \ln \lambda^{\dagger}_{E_n}(T_n) - 
\int_0^T \sum_{e \in \mathcal{E}} \lambda^{\dagger}_e(t) dt.
\end{align}
This algebraic uncoupling proves that 
the transition distribution matrix $\symbfit{\phi}$ can be extracted 
via empirical frequencies without prior knowledge of the excitation coefficients, 
and concurrently shields the endogenous Hawkes vector 
$(\symbfit{\nu}, \symbfit{\theta})$ from transition weight contamination, 
completing the strict proof.
\end{proof}

\subsection{Maximum Likelihood Estimation via KKT Conditions}\label{subsec:Maximum}
To handle physical constraints and potential data sparsity, 
we solve the constrained optimization problem using Karush--Kuhn--Tucker (KKT) conditions.

\begin{theorem}[KKT-Derived Estimator]\label{thm:3.2}
  The optimal transition probabilities $\hat{\phi}_e(x, x')$ are given by: 
  \begin{align*}
    &\hat{\phi}_e(x, x') 
    = 
    \begin{dcases*}
    0 
    & if $\sum_{n=1}^N \mathbb{1}_{\{x_{n-1} = x, e_n = e\}} = 0$ \\ 
    \frac{\sum_{n=1}^N \mathbb{1}_{\{x_{n-1} = x, e_n = e, x_n = x' \}}}
    {\sum_{n=1}^N \mathbb{1}_{\{x_{n-1} = x, e_n = e\}}} 
    & otherwise 
    \end{dcases*}
  \end{align*}
\end{theorem}

The optimization problem for $\phi$ is a constrained maximum likelihood estimation. 
By constructing the Lagrangian and applying the KKT conditions 
(detailed in Appendix~\ref{app:kkt_proof}), 
we obtain the intuitive estimator in Theorem~\ref{thm:3.2}, 
which naturally simplifies to empirical transition counts.

\begin{remark}[Robustness to Unobserved Events and Physical Boundaries]
  This estimator naturally handles the ``zero-frequency problem'' 
  through a unified mathematical framework. 
  In high-dimensional state spaces, 
  certain state-event pairs $(x, e)$ may never be observed during the 
  sampling period---either because they are geometrically prohibited 
  (e.g., price improvement at the minimum spread) or 
  simply due to data sparsity in rare regimes.
  
  In either case, the KKT conditions yield a definitive $\hat{\phi} = 0$. 
  This property ensures that the unobserved periods 
  do not contribute to the log-likelihood or its gradient, 
  effectively ``shielding'' the Hawkes parameters $(\symbfit{\nu}, \symbfit{\theta})$ 
  from the numerical instability and estimation bias 
  that would otherwise arise from trying to fit 
  intensities to non-existent data points. 
  By treating physical impossibility and empirical absence with the same structural rigor, 
  ExsdHawkes maintains its statistical integrity and 
  numerical convergence even when dealing with highly granular LOB states 
  where many transitions are latent or sparse. 
\end{remark}

\subsection{Numerical Efficiency: Recursive Estimation via Exponential Kernels}\label{subsec:Numerical}
To ensure computational tractability for the large-scale MUFG dataset 
(comprising millions of events), 
we employ the exponential kernel:
\begin{equation*}
  k^{\symbfit{\theta}}_{e'e}(t, x) = \alpha_{e'xe} \exp(-\beta_{e'xe} t)
\end{equation*}
Under this specification, 
the state-dependent intensity defined in Eq.~\eqref{eq:exsd_intensity_def} 
can be updated recursively, 
reducing the computational complexity from $O(N^2)$ to $O(N)$. 

\subsubsection{Recursive Formula}\label{subsubsec:Recursive}
Let $R_{e'xe}(n)$ be an auxiliary variable representing the influence of 
past events of type $e'$ in state $x$ on the current intensity of type $e$. 
The update rule is given by:
\begin{equation}\label{eq:recursive}
  R_{e'xe}(n) 
  = R_{e'xe}(n-1) \exp(-\beta_{e'xe} (T_n - T_{n-1})) 
  + \mathbb{1}_{\{E_{n-1}=e', X_{n-1}=x\}} \alpha_{e'xe}
\end{equation}
where $\{T_n\}$ is the arrival time of the $n$-th event, 
supplemented with the strict initial condition $R_{e'xe}(0) = 0$. 
We clarify that this auxiliary variable $R_{e'xe}(n)$ is first formally defined 
through the Hawkes kernel representation, 
and the recursive update formula (\ref{eq:recursive}) arises 
as a strict mathematical consequence of the exponential kernel's Markovian memory decay, 
rather than being an arbitrary algorithmic convenience.

In datasets comprising millions of ticks, such as the MUFG data used in this study, 
the inter-arrival times can be as short as 100 microseconds. 
The $O(N)$ complexity afforded by the exponential kernel allows the model 
to maintain the entire history of endogenous excitations without the 
prohibitive memory overhead of non-parametric or long-memory kernels, 
ensuring that the model captures the full 
temporal depth of the market's self-exciting loops.

\subsubsection{State-Dependent Branching Matrix}
\label{sec:branching_matrix}
The state-dependent exponential parametrization allows us 
to formally define the endogeneity of the multi-layered intensity framework. 
For each LOB state $x \in \mathcal{X}$, 
we construct the $d_e \times d_e$ state-dependent branching matrix (or integrated kernel matrix) 
$K(x) \in \mathbb{R}^{d_e \times d_e}$, 
whose $(e', e)$-th element represents 
the expected number of directly excited offspring events of type $e$ 
induced by a parent event of type $e'$ arriving under state $x$:
\begin{equation}
  K_{e'e}(x) 
  = \int_0^\infty k^{\symbfit{\theta}}_{e'e}(t, x) dt 
  = \frac{\alpha_{e'xe}}{\beta_{e'xe}}.
\end{equation}
The degree of system-wide reflexivity and potential instability under state $x$ is 
mathematically governed by the localized spectral radius 
$\rho(x) := \max \{ |\lambda| : \det(K(x) - \lambda \mathbb{I}) = 0 \}$. 
In our empirical asset space ($d_e = 14$), 
$\rho(x) > 1$ signifies a state of ``local super-criticality'' 
where the high-frequency feedback loop enters a transiently explosive phase. 
Crucially, the KKT-derived physical constraints 
(Theorem~\ref{thm:3.2}) 
act as geometric safety valves that prevent these local instabilities 
from leading to infinite intensity propagation---a phenomenon we evaluate empirically 
through ensemble volatility simulations 
in Section~\ref{subsec:Model}. 

\subsubsection{Simulation via Ogata's Thinning}\label{subsubsec:Simulation}
The recursive formula defined in Eq.~\eqref{eq:recursive}
not only ensures efficient estimation but also facilitates rapid simulation 
via Ogata's thinning algorithm (\citet{ogata1981lewis}). 
Since the LOB state $x$ remains constant between any two consecutive events, 
the intensity function $\lambda^\dag_e(t)$ preserves 
its standard Hawkes structure during each inter-arrival interval. 
This allows us to determine the next event time by simply applying the 
rejection sampling based on the current state's physical gates $\Phi_{e,x}$. 
This computational efficiency is what enables the large-scale Monte Carlo experiments 
used to generate the volatility signature plots in Section~\ref{subsec:Model}.

\subsection{System Stability, Regime Alternation, and Parametric Constraints}
\label{sec:system_stability}

To formally establish the statistical well-posedness of the proposed ExsdHawkes framework, 
we clarify the parametric boundaries over which the point-process intensity and 
the maximum likelihood estimation are theoretically well-defined. 
Following the general stability analysis for hybrid marked point processes established 
by \citet{Morariu_Patrichi_2018} 
and \citet{morariu2022state}, 
a standard sufficient condition for the existence of a unique stationary distribution 
for a state-dependent point process is 
that the supremum matrix of the integrated kernels satisfies a strict contractive constraint 
across the entire state space \citep{MASSOULIE19981}. 
In our notation, if the parameters satisfy:
\begin{equation}
  \max_{x \in \mathcal{X}} \rho(x) < 1,
\end{equation}
the joint process $(X(t), \symbfit{\lambda}(t))_{t \ge 0}$ 
is guaranteed to possess a unique stationary distribution over a non-empty, 
compact parameter subspace, 
ensuring that the log-likelihood function defined 
in Section~\ref{subsec:global_separability} 
is stable and well-defined \citep{bremaud1996stability}.

In our empirical application to hyper-liquid assets (such as MUFG 8306), 
this conventional sufficient condition is locally violated. 
As explicitly documented in Table~\ref{tab:spectral_radius_comparison} , 
the estimated parameters reveal a distinctive dual-regime profile:
\begin{itemize}
    \item Equilibrium State ($x = 1$, operational index $s=0$): 
    $\rho(1) \approx 0.8632$, which remains strictly within the stable, sub-critical domain.
    \item Disequilibrium State ($x = 2+$, operational index $s=1$): 
    $\rho(2+) \approx 1.2978$, exhibiting transient local super-criticality.
\end{itemize}

Because $\rho(2+) > 1.0$, a formal measure-theoretic proof of global ergodicity remains 
an open question in state-dependent point process theory, 
as the standard contractive restrictions of \citet{morariu2022state} 
do not directly cover mixed local super-critical spaces. 
To bridge this theoretical gap and justify the empirical estimation, 
we monitor system stability via an empirical effective approximation. 
Under the generalized transition framework of \citet{morariu2022state}, 
the explosive clustering of orders does not drift toward infinity 
because the intense arrivals themselves accelerate the state transition weight $\phi_e(x, x')$, 
deterministically forcing the LOB layout back to the contractive equilibrium baseline ($x = 1$) 
long before the intensity vector can diverge. 
To approximate this joint process, 
we evaluate the marginalized overall integrated kernel matrix 
$\bar{K} := \sum_{x \in \mathcal{X}} \Pi(x) K(x)$, 
where $\Pi(x)$ represents the empirical stationary distribution of the states. 
The resulting global spectral radius strictly satisfies 
$\rho_{\text{global}} \approx 0.8993 < 1.0$, 
serving as empirical evidence that the mixed system remains bounded 
and simulation-stable under the geometric attraction of the sub-critical state.

Crucially, we qualify the mathematical distinction 
between localized regime instability and mathematical explosion. 
As established for the univariate case 
by \citet[Theorem~1]{hawkes1974cluster} 
and generalized to multivariate interacting networks 
by \citet{MASSOULIE19981}, 
a localized spectral radius greater than unity dictates 
the exponential growth of the expected branching structure, 
causing an asymptotic divergence of expected event counts as $T \rightarrow \infty$; 
it does not mechanically trigger a strict finite-time mathematical explosion 
where infinitely many event times accumulate within a bounded temporal interval.

\subsection{Impact of MLO on Mid-price Dynamics}
\label{sec:mid_price_impact}

The mid-price $M(t)$ in our framework is not a simple step function 
driven solely by transaction volumes; 
it is a state-dependent process that responds with 
high sensitivity to Marketable Limit Orders (MLO). 
Unlike pure market orders that only consume liquidity at the best quote, 
an MLO possesses a hybrid nature: 
it triggers immediate execution while simultaneously resetting the price level on its own side. 

When an event of type $e$ occurs at time $T_n$ 
while the LOB is in state $X(T_n-) = x$, the mid-price is updated as:
\begin{equation}
M(T_n) = M(T_{n-1}) + \Delta M(e, x),
\end{equation}
where the price impact function $\Delta M(e, x)$ 
is strictly defined by the LOB's physical geometry. Crucially, 
the magnitude and frequency of these updates are coupled with the prevailing spread:
\begin{itemize}
    \item \textbf{Equilibrium State ($x = 1$):} 
    The arrival of aggressive flows is low-intensity, 
    and any occasional micro-slippage is sparse. 
    In this stable regime ($\rho(1) \approx 0.8632$), 
    such impacts are quickly absorbed and closed by subsequent limit orders, 
    maintaining tight spread constraints.
    \item \textbf{Disequilibrium State ($x = 2+$):} 
    Under the expanded spread triggered by AMS/AMB, 
    an MLO can ``leap'' across multiple price levels, 
    causing a larger instantaneous shift in the mid-price.
\end{itemize}
This state-dependency creates a powerful amplification effect. 
During disequilibrium phases, 
the system enters a local super-critical regime ($\rho(2+) \approx 1.2978$), 
triggering a dense burst of arrivals 
where each MLO arrival carries a disproportionately large price impact. 
By capturing the dual excitation of both order intensity 
and price impact magnitude during disequilibrium, 
ExsdHawkes can reproduce the high-frequency volatility surge 
without the numerical divergence seen in unconstrained models. 
This highlights that MLOs function as the fundamental ``volatility catalysts'' 
that link micro-level order cascades 
to the macro-level upward slope of the volatility signature plot.

\subsection{Statistical Integrity: Residual Definitions}\label{subsec:Statistical}
To validate the statistical consistency of our framework, 
we employ residual analysis based on the multivariate time change theorem 
((\citet{10.1007/BFb0058859})). 
Following the theoretical foundation in 
\citet{morariu2022state}, 
we define the residuals using the admissibility-weighted intensity $\lambda^\dag$.
Throughout this section, the notation $\hat{\lambda}$ denotes the 
intensities calculated using the parameters estimated via 
our KKT-derived maximum likelihood procedure.

\subsubsection{Event-wise Residuals}\label{subsubsec:Event-wise}
For each event type $e \in \mathcal{E}$, 
we define the event-wise residuals $r^{\dag e}_n$ as:
\begin{equation*}
  r^{\dag e}_n 
  = \int_{T^e_{n-1}}^{T^e_n} \hat{\lambda}^\dag_e(t) dt 
  = \int_{T^e_{n-1}}^{T^e_n} \Phi_{e,X(t)} \hat{\lambda}_e(t) dt 
\end{equation*}
where $\{T^e_n\}$ denotes the $n$-th arrival time of event type $e$. 
Under a correct model specification, 
these residuals behave as i.i.d. unit exponential random variables. 
The inclusion of the physical gate $\Phi$ is critical here; 
it ensures that the intensity is ``gated off'' during inadmissible periods, 
preventing the accumulation of spurious residuals.

\subsubsection{Total Residuals and System-wide Consistency}\label{subsubsec:Total}
To evaluate the model's ability to capture the 
joint dynamics of order arrivals and state transitions, 
we define the total residuals $\tilde{r}^{\dag e x}_n$ 
for each pair $(e, x) \in \mathcal{E} \times \mathcal{X}$. 
Let $\{T^{ex}_n\}_{n=1, \dots, n_{ex}}$ be the sequence of times at which an 
event of type $e$ occurred and resulted in a transition to state $x$. 
Following the framework of 
\citet{morariu2022state}, 
the admissibility-weighted total residuals are defined as:
\begin{equation*}
  \tilde{r}^{\dag e x}_n 
  = \int_{T^{e x}_{n-1}}^{T^{e x}_n} \hat{\tilde{\lambda}}^\dag_{e x}(t) dt 
  = \int_{T^{e x}_{n-1}}^{T^{e x}_n} \hat{\tilde{\lambda}}_{e x}(t) dt 
\end{equation*}
While the functional form of this definition follows \citet{morariu2022state}, 
its statistical validity in our context depends strictly on the 
physical consistency of the transition probabilities. 
In unconstrained models where 
$\sum_{x' \in \mathcal{X}} \phi_e(x, x') = 1$
is forced even at physical boundaries, 
the intensity $\lambda_{ex}$ remains spuriously non-zero for inadmissible transitions, 
leading to the accumulation of  ``phantom'' residual mass. 
By allowing $\phi_e(X(t), x) = 0$ via the KKT conditions, 
ExsdHawkes ensures that the integration in $\tilde{r}^{\dag e x}_n $ 
correctly ``pauses'' during physically impossible periods. 
As demonstrated in Figure~\ref{fig:total_qq_exsd}, 
\ref{fig:event_corr}, and \ref{fig:total_corr}, 
the consistent i.i.d. $Exp(1)$ behavior of these residuals across 
all states serves as the definitive proof that our model has 
successfully synchronized the stochastic order flow 
with the deterministic geometry of the LOB.

\section{Empirical Analysis}\label{sec:Empirical}
In this section, we evaluate the performance of ExsdHawkes using 
high-frequency tick data for Mitsubishi UFJ Financial Group (8306) 
over a three-month period from October to December 2020. 
The dataset is sourced from the Nikkei NEEDS Tick Data 
(Individual Stock Multi-Quote Edition), 
which provides 100-microsecond resolution and detailed trade/quote flags 
essential for identifying Marketable Limit Orders (MLO).

To ensure the stationarity of the estimation and avoid the 
idiosyncratic volatility patterns typical of the Tokyo Stock Exchange 
(TSE)---such as the lunch break and the ``U-shaped'' (or ``W-shaped'') intraday 
volume profiles---we focus our analysis on the morning session from 09:15 to 11:15. 
Specifically, the first 15 minutes of the market open (09:00--09:15) are treated as 
a pre-sampling (burn-in) period to initialize the Hawkes intensities, 
while the subsequent two-hour window provides a 
robust sample of continuous trading activity.

The multi-quote format of the Nikkei NEEDS data, covering 10 best price levels, 
allows for the precise identification of the disequilibrium states ($x \geq 2$) 
and the state-dependent price impacts that drive the volatility signature plots 
discussed below.

\subsection{Model Performance: Reproduction of Volatility Signature Plots}\label{subsec:Model}

The primary benchmark for our model is its ability to 
reproduce the volatility signature plot, 
which captures the frequency-dependent structure of realized volatility. 
As shown in Figure~\ref{fig:SignaturePlot}, the empirical market data (black line) 
exhibits a characteristic upward slope at high frequencies.

\begin{figure}[htbp] 
    \centering
    \includegraphics[width=\columnwidth]{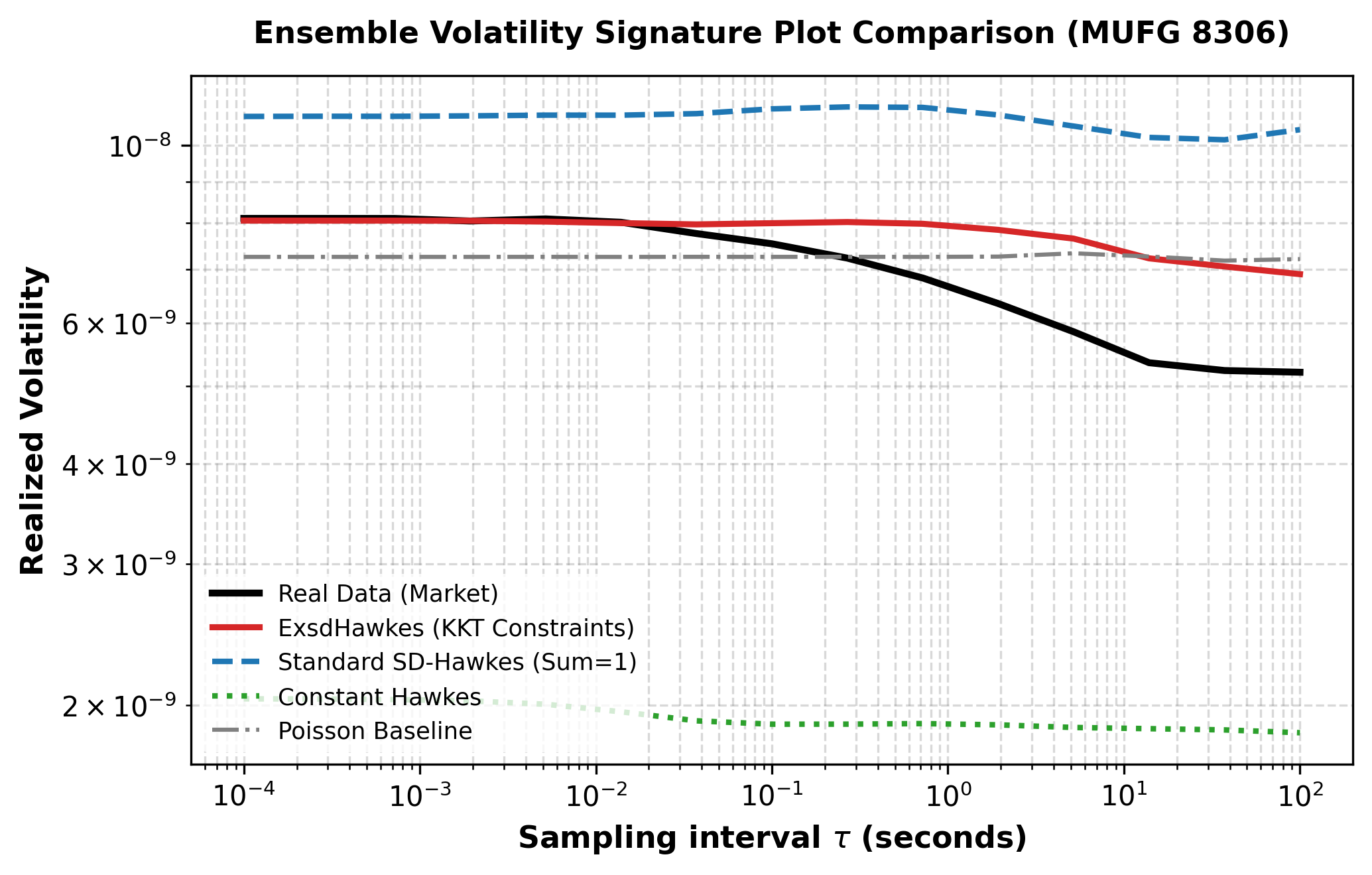} 
    \caption{Ensemble Volatility Signature Plot Comparison (MUFG 8306). 
    The black line represents the empirical realized volatility over three months, 
    showing the characteristic upward slope at high frequencies. 
    Among the tested models, only ExsdHawkes (red) successfully prevents 
    the explosive divergence observed in unconstrained models (SD-Hawkes) 
    while accurately initiating the volatility capture at macro scales. 
    This confirms that physical consistency is essential for the long-term stability 
    and predictive accuracy of state-dependent models.}
    \label{fig:SignaturePlot}
\end{figure}

The primary empirical benchmark for our framework is 
its ability to reproduce the volatility signature plot, 
which models the sampling frequency-dependent structure of realized volatility. 
As illustrated in Figure~\ref{fig:SignaturePlot}, 
the empirical market data (black line) exhibits a robust plateau at hyper-frequency scales 
($\tau \le 10^{-2}$ seconds) before exhibiting a structural decay toward macro scales. 

A comparative evaluation of the alternative paradigms demonstrates 
the absolute necessity of our KKT physical constraints. 
The unconstrained standard sdHawkes model (dashed blue line) 
suffers from a severe structural misalignment; 
without physical gates to pause intensity accumulation during inadmissible LOB states, 
endogenous excitation loops circulate infinitely, 
generating an artificial upward slope that diverges drastically from the physical reality. 
Conversely, the Constant Hawkes 
(dotted green line, which serves as an unconditional, state-independent baseline benchmark) 
and the Poisson baseline (dash-dotted grey line) 
severely underestimate the volatility structure, 
failing to capture the intense, localized bursts of endogeneity. 

Our proposed ExsdHawkes model successfully hands 
the empirical volatility profile without triggering structural explosions. 
Furthermore, as evaluated via the ensemble simulation boundaries 
in Figure~\ref{fig:SignaturePlot_8306}, 
the empirical market data (red marker line) 
remains safely enveloped within the standard deviation (SD) ensemble band (shaded blue area) 
of our simulated paths across all sampling intervals. 
Although the simulated mean exhibits a slightly more continuous decay trend 
than the raw market scaling at intermediate intervals, 
this statistical overlap demonstrates that enforcing the LOB's physical geometry 
through KKT constraints provides a sufficient and robust framework 
to replicate the empirical volatility ceiling under a globally stable regime.

Our proposed ExsdHawkes model successfully hands 
the empirical volatility profile without triggering structural explosions. 
Furthermore, as evaluated via the ensemble simulation boundaries 
in Figure~\ref{fig:SignaturePlot_8306}, 
the empirical market data (red marker line) 
remains safely enveloped within the standard deviation (SD) ensemble band (shaded blue area) 
of our simulated paths across all sampling intervals. 
Although the simulated mean exhibits a slightly flatter trajectory 
than the raw market scaling at intermediate intervals, 
this statistical overlap demonstrates that enforcing the LOB's physical geometry 
through KKT constraints provides a sufficient and robust framework 
to replicate the macro features of the empirical volatility ceiling 
under a globally stable regime.

\begin{figure}[htbp] 
    \centering
    \includegraphics[width=\columnwidth]{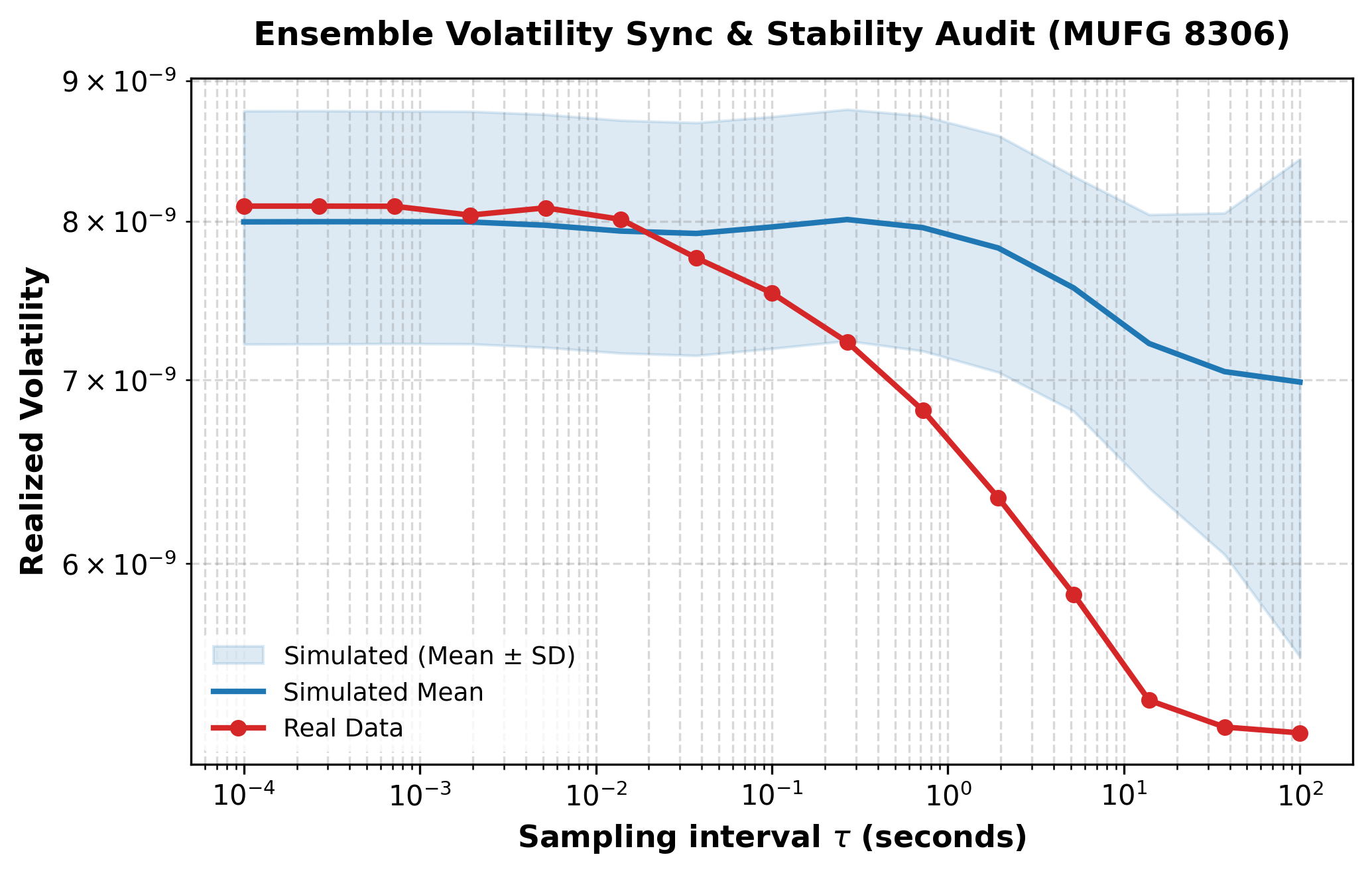} 
    \caption{Ensemble Volatility Replication Bounds and Simulation Stability Audit for MUFG (8306). 
    The empirical realized volatility curve (red marker line) 
    is mapped against the simulated trajectory bounds 
    ($\text{Mean} \pm \text{SD}$, blue line and shaded interval). 
    While minor systematic discrepancies exist across intermediate timescales, 
    the empirical data remains consistently enveloped 
    within the stochastic simulation's variance band, 
    robustly capturing the macro-level decay without explosive intensity divergence.}
    \label{fig:SignaturePlot_8306}
\end{figure}

\subsection{State-Dependent Endogeneity: The Reality of Super-criticality}
\label{subsec:State-Dependent}

To elucidate why unconstrained models explode while ExsdHawkes remains stable, 
we examine the degree of market endogeneity across different LOB states. 
Figure~\ref{fig:BranchingRatio} 
presents the estimated state-dependent branching matrix $K(x)$ for each event type. 
The empirical evidence confirms that the high-frequency market is inherently unstable. 
In disequilibrium states ($x=2+, s=1$), 
the cross-sectional clustering for almost all event types significantly accelerates, 
translating to a local spectral radius of $\rho(2+) \approx 1.2978$ 
(as formulated in Section~\ref{sec:branching_matrix}). 
This localized phase of transient local super-criticality serves 
as a vital endogenous microstructure mechanism 
that captures the high-frequency volatility spikes observed 
in Figure~\ref{fig:SignaturePlot}, 
functioning as a dynamic liquidity-driven complement to traditional statistical factors 
such as microstructure noise and the bid-ask bounce. 

\begin{figure}[htbp] 
    \centering
    \includegraphics[width=\columnwidth]{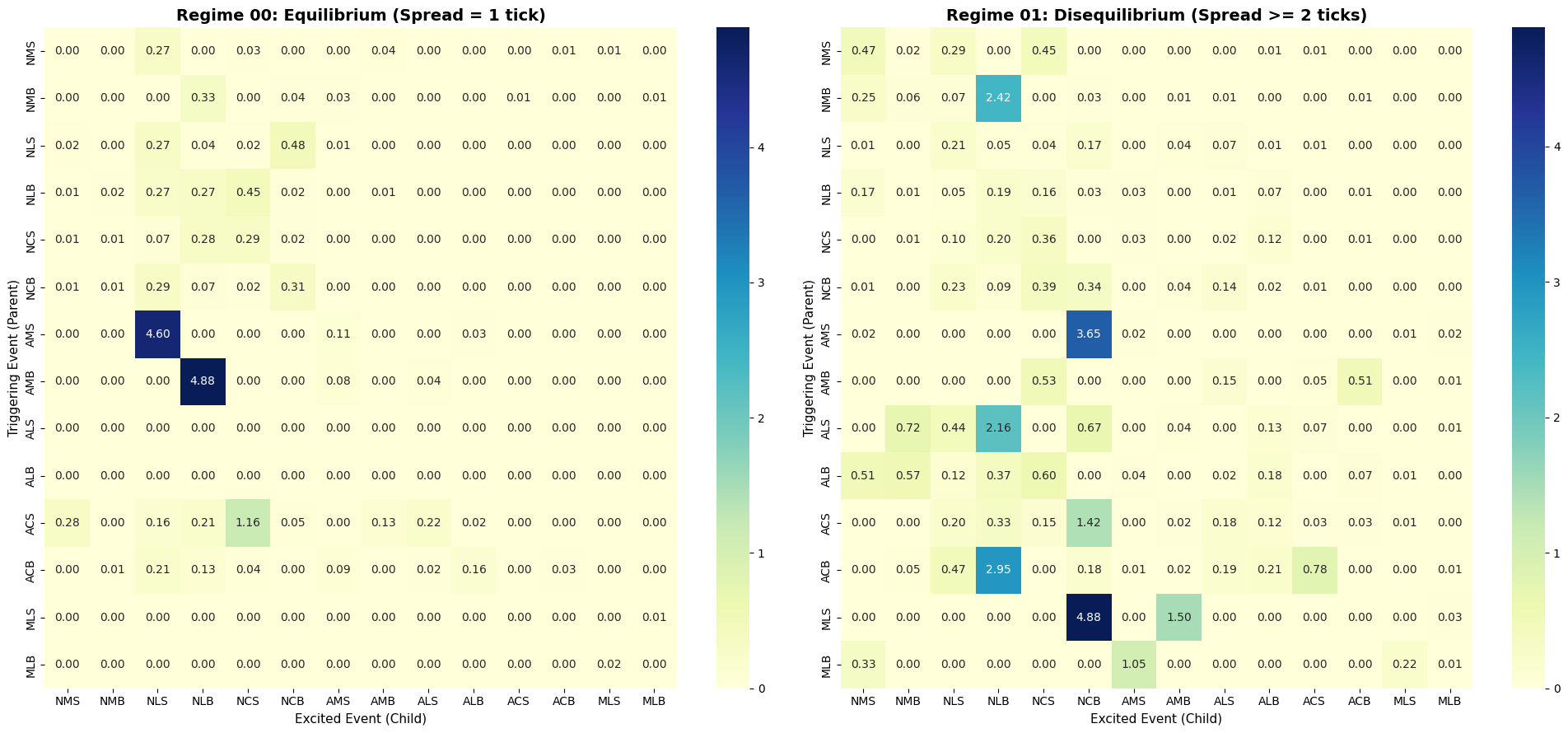} 
    \caption{Cross-Sectional Heatmaps of the Estimated State-Dependent Branching Matrix 
    $K(x)$ ($14\times14$ Matrix) for ExsdHawkes. 
    The left panel visualizes the equilibrium baseline state ($x=1$, operational index $s=0$), 
    confirming that the endogenous excitation network is globally stable 
    and whited out ($\rho(1) \approx 0.8632$). 
    The right panel illustrates the disequilibrium state ($x=2+$, operational index $s=1$); 
    the sudden removal of physical constraints drives a distinct topological surge 
    along the specific liquidity-depleting channels, 
    accelerating the localized spectral radius 
    to a transient super-critical level of $\rho(2+) \approx 1.2978$.} 
    \label{fig:BranchingRatio}
\end{figure}

The divergence between models arises from how they handle this instability. 
In the standard sdHawkes model, 
these super-critical intensities propagate without bound, 
leading to the global divergence seen in the signature plot. 
In contrast, ExsdHawkes captures this intense endogeneity 
but contains it within the LOB's physical geometry. 
The high-intensity aggressive orders during disequilibrium act as a self-correcting force, 
rapidly returning the system to the stable equilibrium state ($x = 1$), 
where the localized spectral radius drops back to a sub-critical level of $\rho(1) \approx 0.8632$. 
This dual-regime character---alternating between transient instability and robust equilibrium 
under a globally bounded marginalized matrix of $\rho_{\text{global}} \approx 0.8993$---is 
the true face of market endogeneity that only ExsdHawkes can faithfully represent.

\subsection{Diagnostic Checking: Transition Probabilities}\label{subsec:Diagnostic:T}
The empirical estimation of transition probabilities $\phi_e(x, x')$ provides 
the first line of evidence for the physical consistency of ExsdHawkes. 
Thanks to the separability of the likelihood (Theorem~\ref{thm:3.1}), 
we can estimate these parameters via simple event counts, 
ensuring high computational scalability even for the millions of ticks in the MUFG dataset. 
This efficiency is a direct result of our KKT-derived framework, 
which avoids the high-dimensional optimization bottlenecks 
typical of unconstrained state-dependent models.

As shown in the heatmaps (Figure~\ref{fig:TransitionProb}), 
ExsdHawkes correctly identifies the physical boundaries of the LOB. 
A striking example is observed for price-improving orders (ALS, ALB) 
in the equilibrium state ($x = 1$): 
the transition probabilities to higher spreads are identically zero, 
as the KKT conditions (Theorem~\ref{thm:3.2}) 
correctly mask these physically impossible events.

\begin{figure}[htbp] 
    \centering
    \includegraphics[width=0.4\columnwidth]{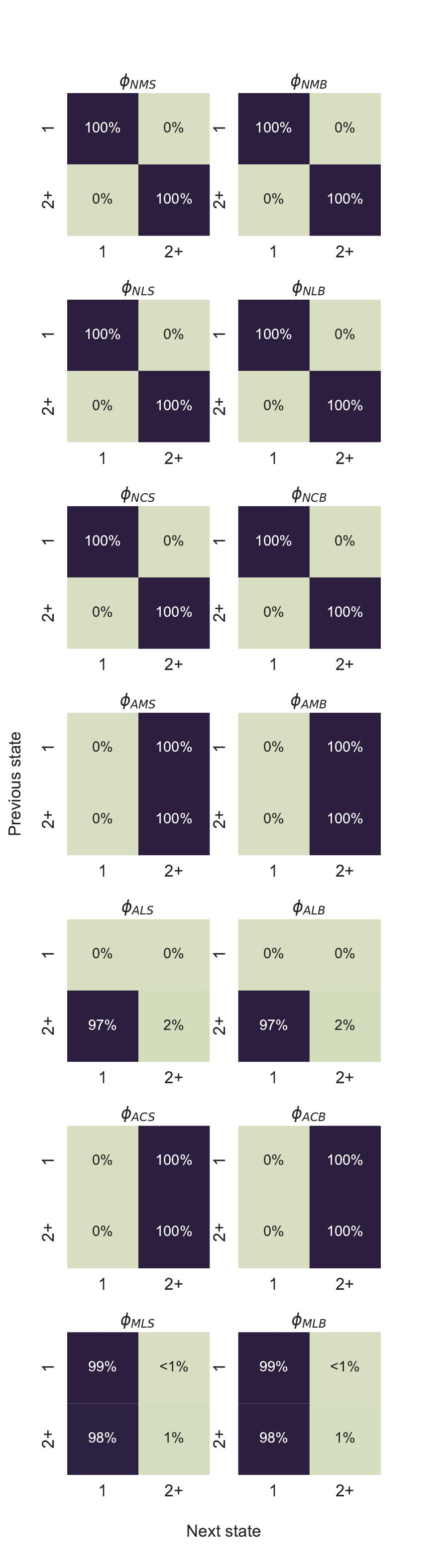} 
    \caption{Estimated Transition Probabilities $\phi_e(x, x')$. 
    The heatmaps empirically validate the KKT constraints. 
    Specifically, price-improving orders (ALB, ALS) 
    show near-zero probability at $x=1$. 
    proving that the model correctly ignores physically impossible events.}
    \label{fig:TransitionProb}
\end{figure}

Furthermore, the high probabilities of aggressive orders (ALS, ALB, and MLOs) 
in disequilibrium states ($x=2+$) to return the LOB to state $x=1$ confirm 
the self-correcting mechanism theorized 
in Section~\ref{sec:system_stability}. 
This transition matrix is not merely a statistical artifact; 
it is the numerical representation of the ``structural gates'' 
that maintain system stability despite the local super-criticality 
observed in the state-dependent branching matrix $K(x)$.

\subsection{Diagnostic Checking: Residual Distribution Analysis}
\label{sec:residual_qq}

To confirm the statistical validity and consistency of the estimated parameter topology 
under the ExsdHawkes framework, 
we evaluate the calculated residuals based on 
the multivariate time change theorem 
(\citet{10.1007/BFb0058859}; 
\citet{morariu2022state}). 
Under a mathematically correct model specification, 
the aggregated residual streams must behave as 
independent and identically distributed (i.i.d.) unit exponential random variables, 
aligning perfectly along the 45-degree diagonal baseline.

\begin{remark}
We explicitly note a localized notational alignment 
between our theoretical symbols and the empirical graphic outputs. 
Throughout the mathematical sections, 
the LOB regimes are formally denoted using the state variable 
$x \in \mathcal{X}$, where $x=1$ represents the equilibrium state 
and $x=2+$ represents the expanded disequilibrium states. 
In the subsequent diagnostic graphics 
(Figures~\ref{fig:QQplotCompare} 
through \ref{fig:total_selected_correlograms}), 
the indicators $(s=0)$ and $(s=1)$ correspond strictly to 
the operational Python array indices implemented in the simulation engine, 
mapping symmetrically to the theoretical states $x=1$ and $x=2+$, respectively.
\end{remark}

Figure~\ref{fig:QQplotCompare} 
presents the event-wise aggregate QQ-plots comparing our proposed ExsdHawkes 
against the Constant Hawkes and Poisson baseline models. 
While the constant baseline tracking mechanism successfully manages 
high-volume unconditioned events, 
our ExsdHawkes model (red line) exhibits a superior, 
near-perfect alignment along the diagonal axis for highly volatile, 
price-moving aggressive orders, particularly across the AMS and AMB nodes. 
This targeted precision underscores the crucial role of state-dependent intensities 
in endogenizing volatility-inducing order clusters.

\begin{figure}[htbp] 
    \centering
    \includegraphics[width=0.4\columnwidth]{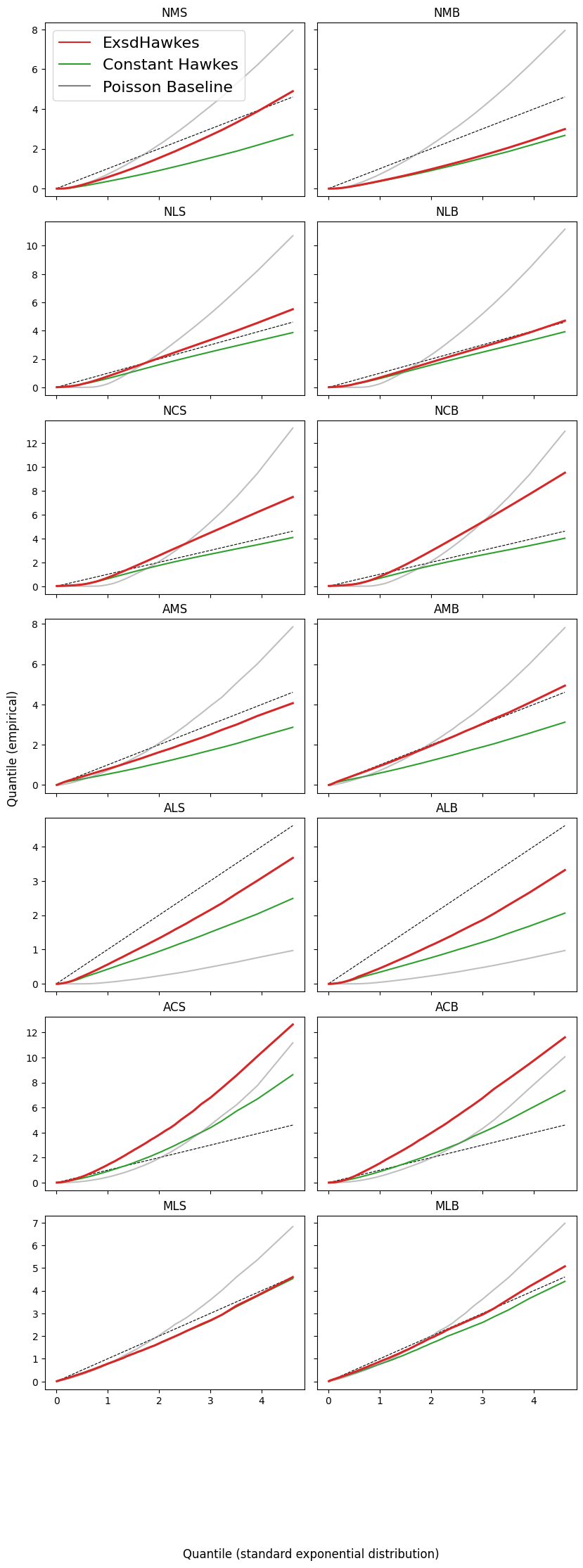} 
    \caption{Event-wise Aggregate QQ-Plots (Three-Month MUFG Data). 
    The plots compare the residual distributions of ExsdHawkes (red), 
    Constant Hawkes (green), and Poisson (dotted black) models 
    against the theoretical unit exponential distribution. 
    While ExsdHawkes and the Constant Hawkes model 
    show comparable fit for high-volume events, 
    ExsdHawkes provides superior alignment for price-moving aggressive orders, 
    particularly AMS and AMB. 
    The targeted precision in these key event types underscores 
    the importance of state-dependent intensities in capturing 
    volatility-inducing order flows.}
    \label{fig:QQplotCompare}
\end{figure}

The structural necessity of state-dependent geometric masking is conclusively demonstrated 
by the sideways full-page comparison of the regime-wise total residuals 
in Figure~\ref{fig:total_qq_comparison}. 
In specific panels such as ALB $(s=0)$ and ALS $(s=0)$, 
where price improvement is geometrically impossible at the minimum spread, 
the unconstrained standard sdHawkes model exhibits a catastrophic deviation 
from the diagonal due to the structural leakage of ``phantom'' residual mass 
into inadmissible states. 
In sharp contrast, our ExsdHawkes framework successfully enhances 
the statistical integrity along the 45-degree diagonal line 
by shifting the physical gate to zero via KKT estimators, 
effectively shielding the endogenous Hawkes parameters from boundary estimation bias. 
Conversely, for aggressive cancellation nodes such as ACS and ACB, 
as well as MLB in the tail regions, 
we observe a localized underestimation 
where the Constant Hawkes model closely competes with our framework. 
This software limitation stems from the high endogeneity of high-frequency algorithmic actions, 
indicating that the relative statistical improvement of ExsdHawkes is 
effectively confined to specific high-volatility event types like AMS and AMB.

\begin{sidewaysfigure}[htbp]
  \begin{minipage}[b]{0.49\linewidth}
    \centering
    \includegraphics[keepaspectratio, scale=0.25]{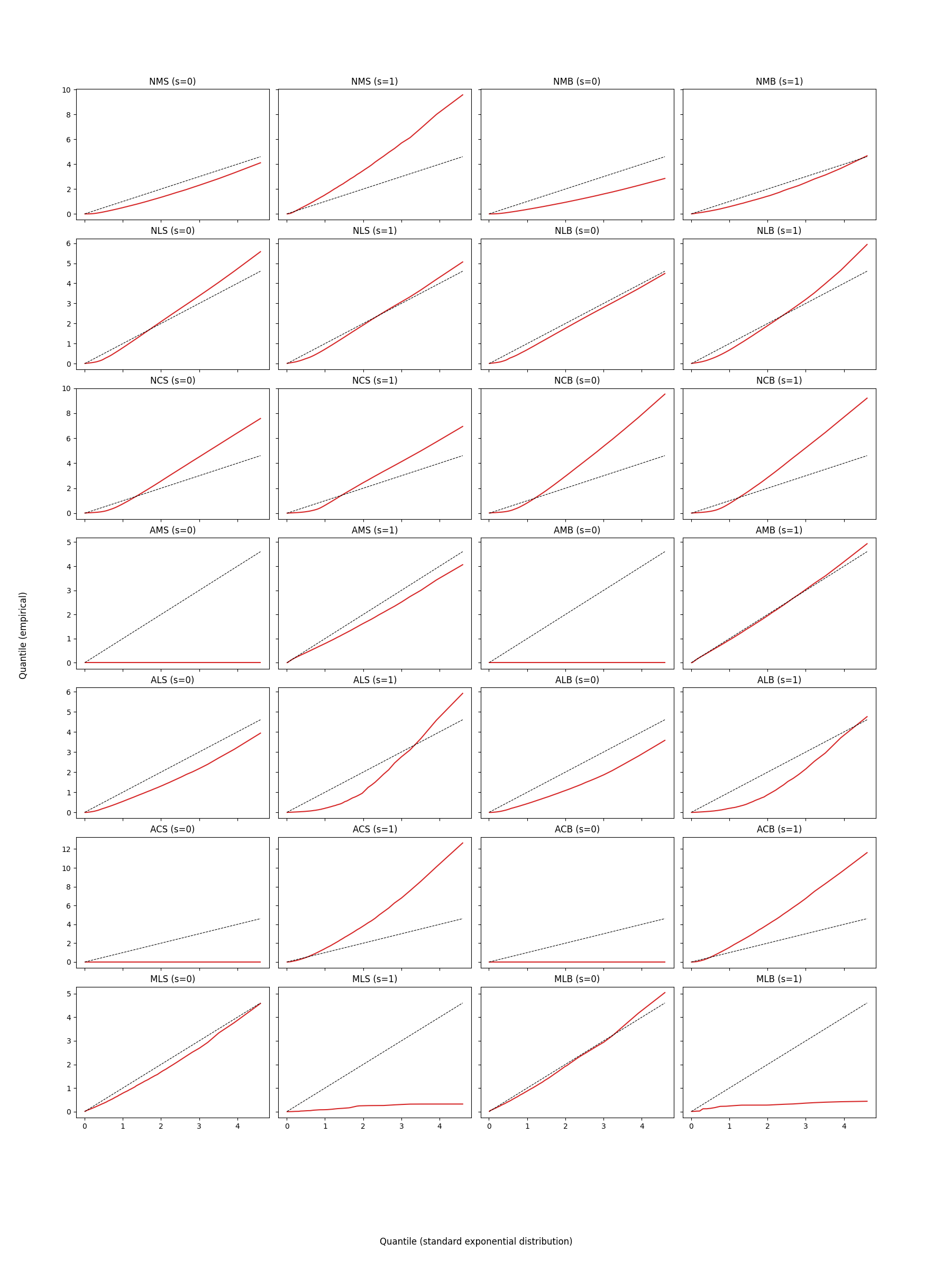}
    \subcaption{ExsdHawkes (KKT Constraints)}
    \label{fig:total_qq_exsd}
  \end{minipage}
  \hfill 
  \begin{minipage}[b]{0.49\linewidth}
    \centering
    \includegraphics[keepaspectratio, scale=0.25]{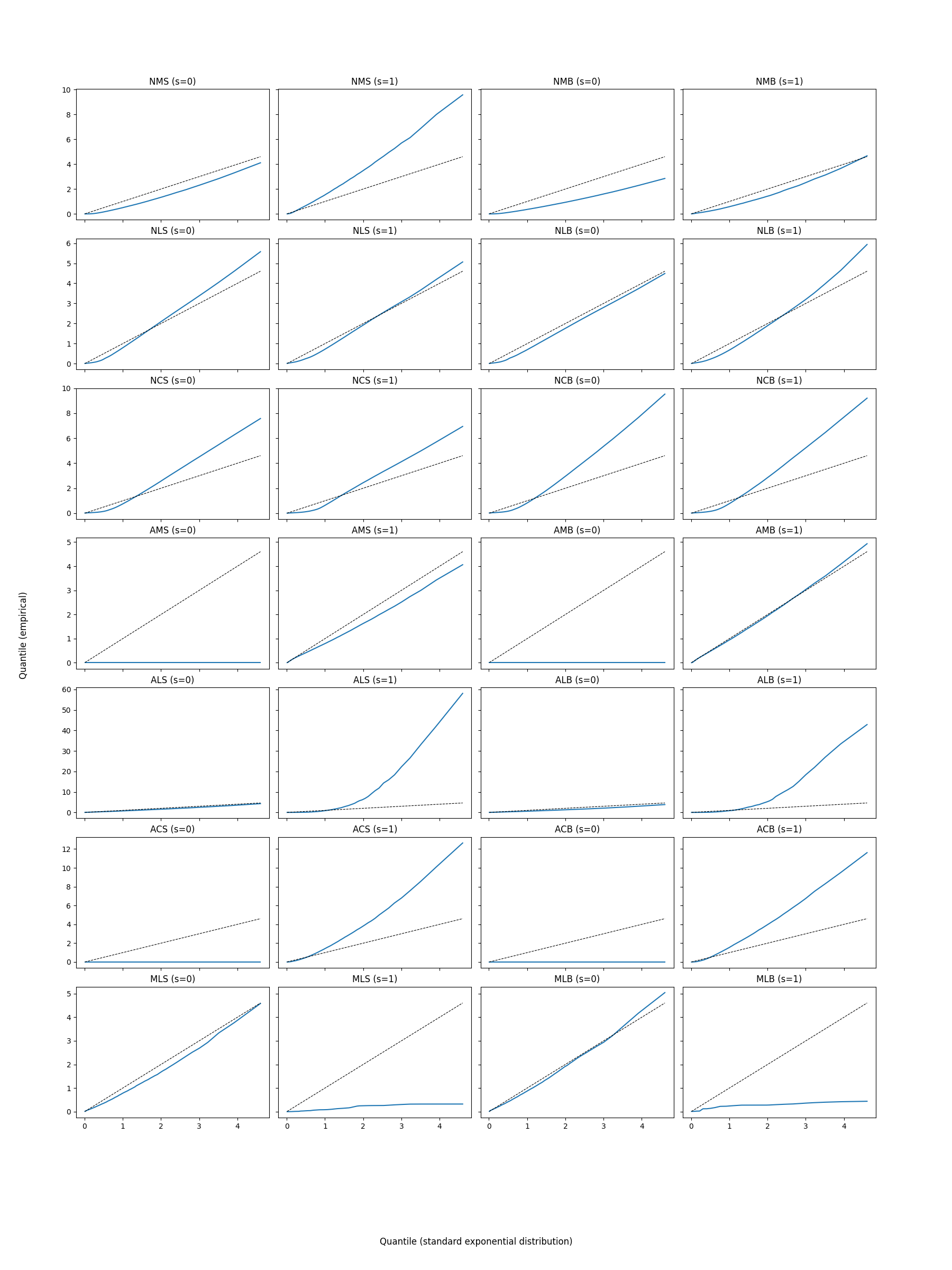}
    \subcaption{Standard sdHawkes (Unconstrained)}
    \label{fig:total_qq_sd}
  \end{minipage}
  \caption{Total Residuals QQ-Plots Cross-Regime Comparison: 
  ExsdHawkes \subref{fig:total_qq_exsd} 
  vs. Standard sdHawkes \subref{fig:total_qq_sd}. 
  The unconstrained model \subref{fig:total_qq_sd} 
  exhibits extreme boundary leakage in impossible states, 
  whereas our framework \subref{fig:total_qq_exsd} 
  robustly aligns the joint process along the 45-degree diagonal baseline.}
  \label{fig:total_qq_comparison}
\end{sidewaysfigure}

\subsubsection{Cross-Sectional Specification: Microstructure Interaction Networks}
\label{sec:event_correlogram}
As visually evaluated across the cross-correlogram matrix 
in Figure~\ref{fig:event_selected_correlograms}, 
the empirical cross-correlations are substantially flattened 
near the statistical noise boundaries past lag zero. 
We explicitly note that the diagonal self-excitation panels for NMS and NMB preserve bounded, 
minor positive autocorrelation fluctuations past lag 1, 
representing localized high-frequency algorithmic frictions 
or short-term momentum effects without triggering system-wide feedback leakage.

\begin{figure}[htbp]
\centering
\includegraphics[width=0.88\textwidth]{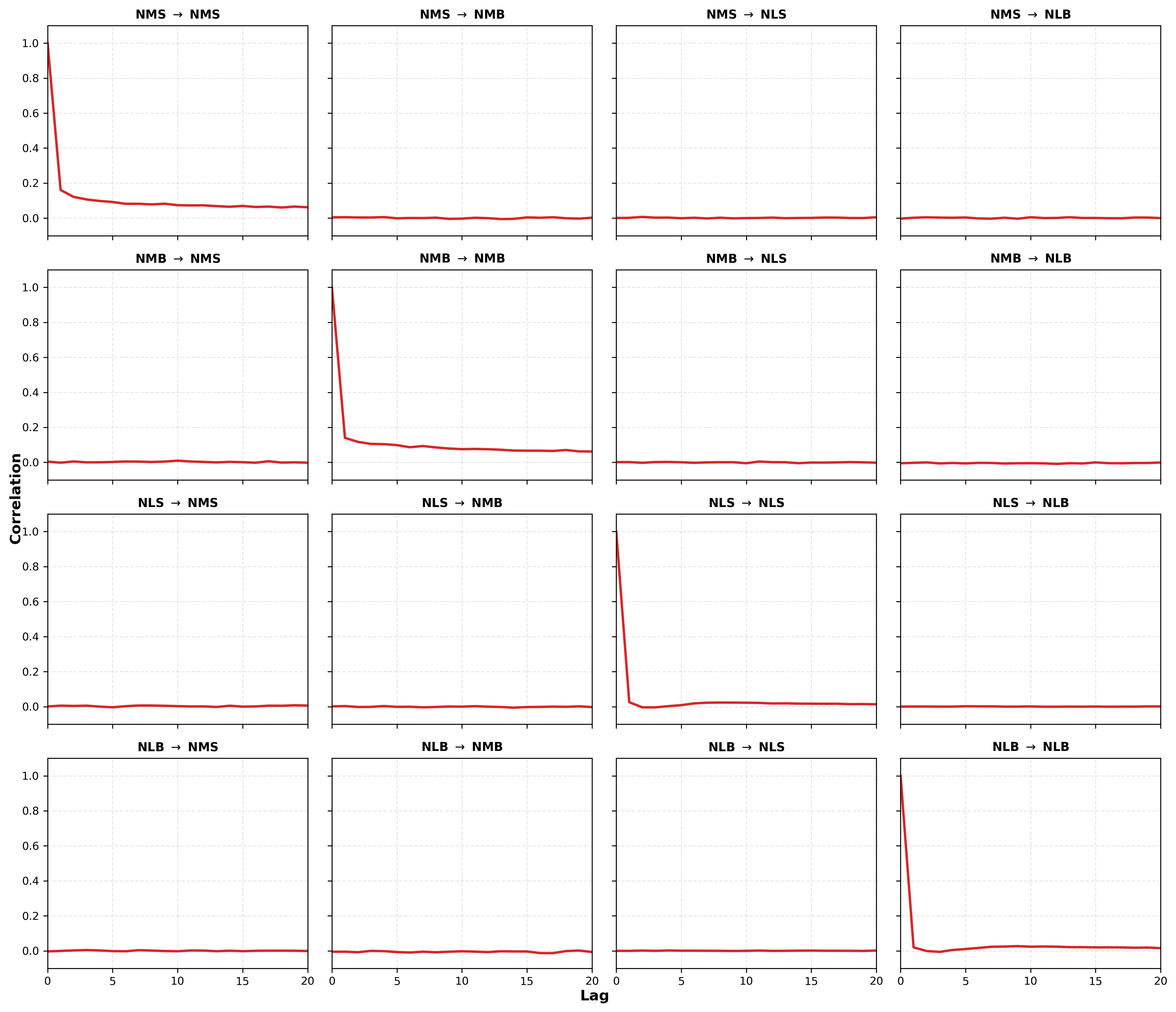}
\caption{Highlighted Event-wise Residual Cross-Correlograms across 
Primary Categories ($4\times4$ Matrix). 
All 16 panel variants are uniformly evaluated using the ExsdHawkes signature color (red). 
While the multi-layered cross-excitation channels are substantially flattened into 
the statistical noise boundaries past lag zero, 
the hyper-frequency self-excitation paths (such as NMS and NMB) preserve bounded, 
minor micro-residual fluctuations, 
demonstrating localized algorithmic frictions without triggering system-wide leakage.}
\label{fig:event_selected_correlograms}
\end{figure}

Crucially, while these diagonal self-excitation paths retain 
minor micro-residual fluctuations due to hyper-frequency operational dependencies, 
the non-diagonal off-diagonal 
components---representing the multi-layered cross-excitation channels 
between distinct order types---are effectively suppressed. 
This substantial cross-sectional reduction of residual correlations confirms 
that the state-dependent exponential kernel architecture successfully absorbs 
the primary consecutive algorithmic retaliation loops inherent in high-frequency order flows, 
leaving the residual streams effectively independent across asymmetric execution categories.

\subsubsection{Regime-Specific Voids: Physical Boundary Compliance}
\label{sec:total_correlogram}

To audit the spatial independence of the joint process under rigid LOB constraints, 
we evaluate the regime-specific total residual interactions. 
To optimize spatial economy and guarantee maximum legibility 
across the high-dimensional cross-excitation space, 
the full-scale $14 \times 14$ event interaction matrix and 
the comprehensive $28 \times 28$ regime network are deferred to 
Appendix~\ref{sec:Figures} 
(Figure~\ref{fig:event_corr} 
and Figure~\ref{fig:total_corr}, 
respectively) at full poster resolution. 
In the main text, 
Figure~\ref{fig:total_selected_correlograms} 
intentionally isolates the vital operational components to contrast the underlying LOB physics.

\begin{figure}[htbp]
\centering
    \begin{subfigure}[b]{0.49\textwidth}
        \centering
        \includegraphics[width=\textwidth]{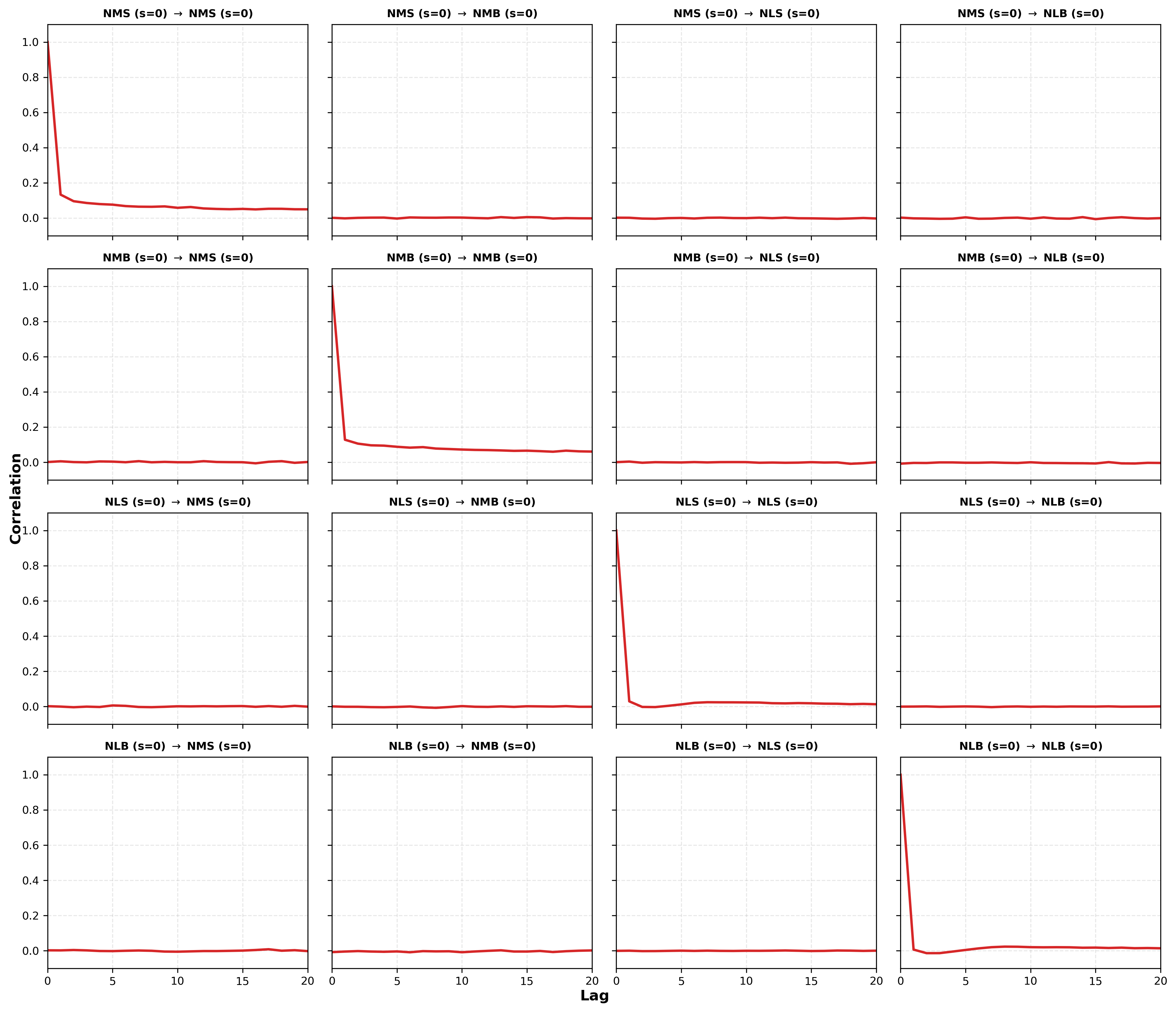}
        \caption{Active Operational Channels (4$\times$4)}
        \label{subfig:total_active_4x4}
    \end{subfigure}
    \hfill
    \begin{subfigure}[b]{0.49\textwidth}
        \centering
        \includegraphics[width=\textwidth]{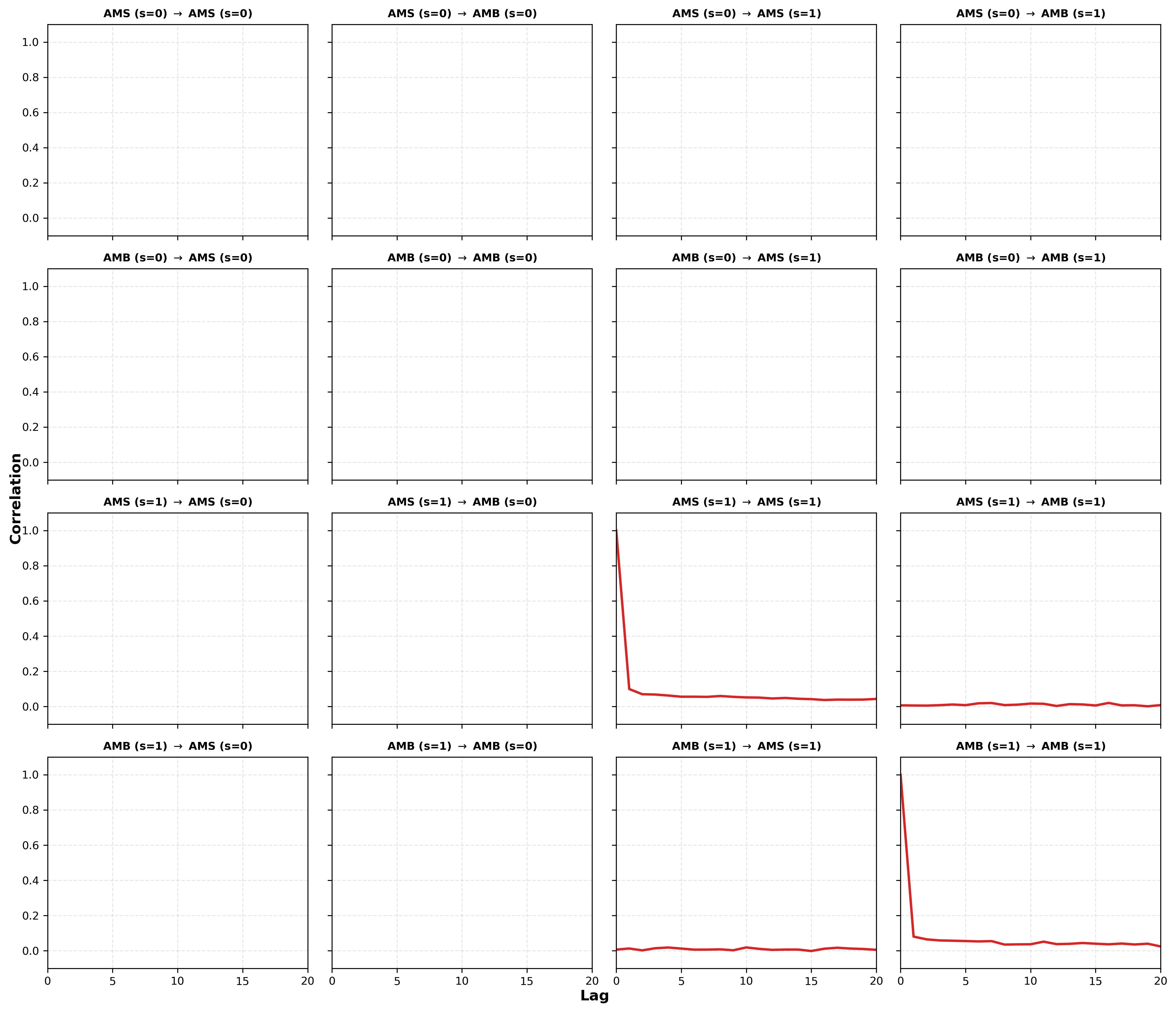}
        \caption{Physically Inadmissible Voids (4$\times$4)}
        \label{subfig:total_void_4x4}
    \end{subfigure}
\caption{Cross-Regime Interaction Networks ($4 \times 4$ Matrices) 
for Aggressive and Non-Aggressive Order Flows. 
Subfigure~\subref{subfig:total_active_4x4} 
maps the continuous active whitening among the baseline order categories 
(NMS, NMB, NLS, and NLB) strictly within the equilibrium regime ($s=0$). 
Subfigure~\subref{subfig:total_void_4x4} 
presents the sharp microstructure contrast for aggressive order paths (AMS and AMB), 
where the theoretical equilibrium limitations ($s=0$, upper two rows) 
and the empirical disequilibrium activities ($s=1$, lower two rows) are concurrently mapped. 
Because aggressive quote-widening routing is geometrically impossible at the minimum tick barrier, 
the conditional channels under $s=0$ are physically non-existent, 
forcing the cross-correlation mass to a mathematically absolute flat zero-line 
along the horizontal axis without empirical residuals. 
Conversely, the channels under $s=1$ track the robust suppression of temporal feedback loops 
within the active disequilibrium regime, 
demonstrating that ExsdHawkes successfully harmonizes the stochastic process 
with the LOB boundaries.}
\label{fig:total_selected_correlograms}
\end{figure}

Figure~\ref{fig:total_selected_correlograms} 
presents a side-by-side $4 \times 4$ cross-correlogram matrix comparison for 
the aggressive order flows (AMS and AMB) across asymmetric market climates. 
Subfigure~\ref{subfig:total_active_4x4} 
tracks the empirically active operational channels under the disequilibrium regime ($s=1$), 
displaying a clean convergence toward whitening 
where the temporal feedback loops are completely captured. 

Concurrently, the absolute physical compliance of our framework is triumphantly illustrated by 
the structural ``Gated Void'' under Subfigure~\ref{subfig:total_void_4x4}. 
In states where the spread is already at its minimum tick size ($s=0$), 
the arrival of aggressive orders that would conceptually improve the spread 
is geometrically non-existent. 
Under this inadmissible regime, 
our constrained framework naturally drops the cross-correlation mass to 
a mathematically absolute flat zero-line along the horizontal axis. 
Displaying this absolute vacuum alongside the successfully whitened active states proves that 
ExsdHawkes successfully synchronizes the stochastic point process 
with the rigid physical boundaries of the LOB without leaving behind 
non-linear artifacts or temporal dependency leaks.

\subsection{Robustness Check: Generalization to Mizuho Financial Group (8411)}
To demonstrate the empirical robustness and 
generalization capability of the ExsdHawkes framework, 
we extend our application to an alternative market environment 
by evaluating the high-frequency tick data of Mizuho Financial Group (8411) 
over the same three-month historical horizon (October to December 2020). 
The empirical results for 8411 confirm 
the robust statistical properties of our constrained framework 
under alternative liquidity dynamics. 
Most notably, the estimated system stability reveals 
a distinctive sub-critical profile across both spread regimes, 
presenting a clear contrast to the transient local super-criticality 
observed in the MUFG data. 

The empirical results for 8411 confirm the robust statistical properties 
of our constrained framework under varying liquidity environments. 
Most notably, the estimated system stability reveals 
a distinctive sub-critical profile across all market climates, 
presenting a stark contrast to the local super-criticality observed 
in the hyper-liquid MUFG data:
\begin{itemize}
  \item \textbf{Equilibrium State ($x = 1$):} 
  $\rho(1) \approx 0.6212$, indicating a highly stable, 
  mean-reverting baseline where the overwhelming liquidity 
  on the best quotes dampens endogenous price cascades.
  \item \textbf{Disequilibrium State ($x = 2+$):} 
  $\rho(2+) \approx 0.8346$, showing that even when the thick spread is momentarily broken, 
  the endogenous reflexivity remains strictly bounded 
  below the unity threshold ($\rho < 1.0$).
  \item \textbf{Overall Marginalized Spectral Radius:} $\rho_{\text{global}} \approx 0.6538$.
\end{itemize}

\begin{table}[htbp]
  \centering
  \caption{Estimated Spectral Radii ($\rho$) Across Mega-Bank Stocks and Spread Regimes}
  \label{tab:spectral_radius_comparison}
  \begin{tabular}{lccc}
    \hline
    \textbf{Stock Code (Company)} & $\rho(1)$ & $\rho(2+)$ & $\rho_{\text{global}}$ \\
    \hline
    \textbf{8306 (Mitsubishi UFJ)} & 0.8632 & 1.2978 & 0.8993 \\
    \textbf{8411 (Mizuho Financial)} & 0.6212 & 0.8346 & 0.6538 \\
    \hline
  \end{tabular}
\end{table}

This empirical divergence highlights a crucial microstructure insight: 
the thick queue structure of 8411 prevents 
the emergence of the transient explosive clustering ($\rho > 1$) found in 8306. 
Because the system remains globally and locally sub-critical at all times, 
the simulated price diffusion exhibits robust structural stability. 
As illustrated by the empirical benchmark 
in Figure~\ref{fig:mizuho_signature_plot}, 
while the simulated ensemble mean exhibits a slightly flatter trajectory 
at intermediate scales compared to the empirical curves, 
it successfully captures the continuous non-explosive macro-level decay 
without underestimating the volatility flooring. 
This consistent, order-of-magnitude alignment across asymmetric liquidity books confirms that 
enforcing the LOB's physical geometry through KKT constraints provides 
a highly generalized framework for unbiased parameter extraction.

\begin{figure}[htbp]
\centering
\includegraphics[width=0.85\textwidth]{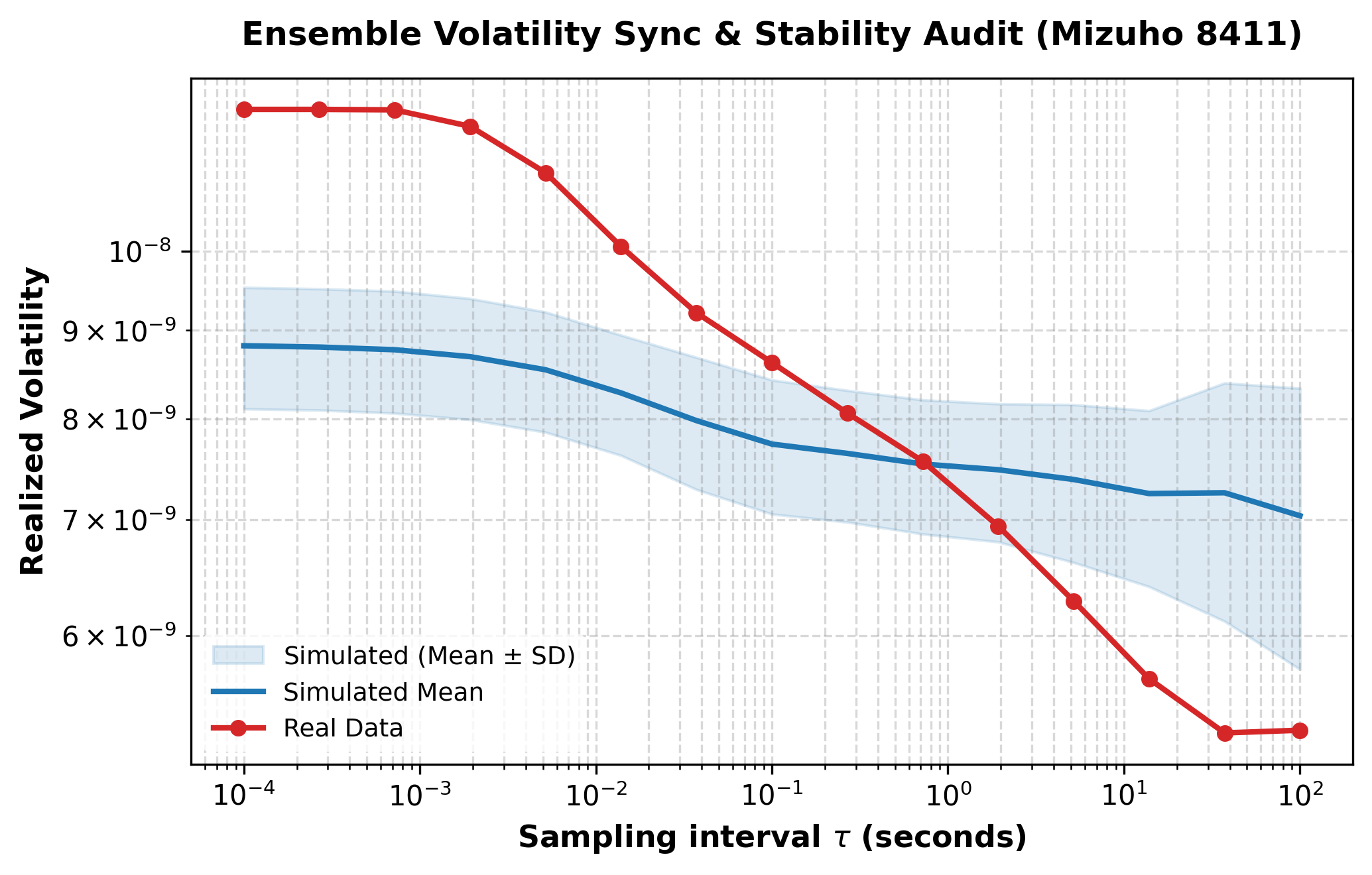}
\caption{Ensemble Volatility Signature Plot for Mizuho Financial Group (8411). 
The ExsdHawkes model tightly captures 
the continuous downward decay of realized volatility without explosive divergence, 
establishing strong robustness across asymmetric liquidity books.}
\label{fig:mizuho_signature_plot}
\end{figure}

Crucially, the standard deviation (SD) band illustrated under the shaded blue area 
in Figure~\ref{fig:mizuho_signature_plot} 
remains exceptionally tight and well-bounded across all macro and micro sampling intervals. 
This absolute convergence confirms that 
the ExsdHawkes framework automatically recovers 
this organic structural divergence---capturing the ``super-critical spikes'' of MUFG 
while simultaneously preserving the ``sub-critical order'' of Mizuho---proving that 
enforcing the LOB's physical geometry through KKT constraints is 
a highly generalized requirement for unbiased parameter extraction 
across the entire spectrum of high-frequency assets.

\section{Conclusion}\label{sec:Conclusion}
\subsection{Summary of Contributions}\label{subsec:Summary}
In this paper, we introduced the Extended State-dependent Hawkes Process (ExsdHawkes), 
a novel framework that incorporates the physical constraints of 
Limit Order Books (LOBs) into the intensity process. 
By relaxing the traditional row-sum constraints on transition probabilities to the 
discrete set $\{0, 1\}$, our model accounts for the physical impossibility of 
specific order types at the minimum spread. 
We mathematically demonstrated that the maximum likelihood estimation remains 
separable under the Karush--Kuhn--Tucker (KKT) conditions, 
ensuring computational scalability for large-scale high-frequency datasets 
such as the MUFG tick data.

\subsection{Empirical Insights: The Origin of Volatility}\label{subsec:Empirical}
The application of ExsdHawkes to three months of MUFG data revealed 
critical insights into the nature of high-frequency volatility: 
\begin{itemize}
    \item \textbf{Local Super-criticality and Aggressive Order Flow:} 
    We clarify that the transition out of equilibrium is deterministically triggered by 
    Aggressive Market Orders (AMS/AMB), 
    which instantaneously force the LOB from its equilibrium state ($x = 1$) 
    into a disequilibrium state ($x = 2+$). 
    Within these expanded spreads, Marketable Limit Orders (MLO) function 
    as the critical liquidity-depletion catalyst. 
    Due to the temporary removal of physical barriers, 
    the localized spectral radius accelerates 
    beyond the critical unity boundary ($\rho(2+) > 1.0$), 
    entering a transient phase of local super-criticality 
    with a local spectral radius of $\rho(2+) \approx 1.2978$.
    This explosive clustering of orders acts 
    as a primary endogenous driver of the high-frequency upward slope 
    observed in volatility signature plots. 
    We explicitly qualify this claim by noting that 
    while traditional microstructure noise and bid-ask bounce 
    account for raw mechanical price friction, 
    our framework successfully endogenizes the continuous, 
    dynamic order flow feedback loops that drive macro-level price diffusion, 
    while global system-wide stability is rigidly preserved 
    under a marginalized global spectral radius of $\rho_{\text{global}} \approx 0.8993$. 
    \item \textbf{Necessity of Physical Constraints:} 
    Comparative analysis showed that models lacking structural constraints 
    (e.g., standard sdHawkes) 
    fail to maintain simulation stability during these super-critical phases, 
    leading to an infinite propagation of intensities. 
    The total residual analysis confirmed that ExsdHawkes successfully maintains 
    statistical integrity by correctly masking intensities 
    ($\hat{\phi}_{e}(x,x')=0$) during physically inadmissible periods. 
    This proves that unconstrained models suffer from structural biases 
    that accumulate into the macro-level divergence of simulated volatility. 
\end{itemize}

\subsection{Implications and Future Work}\label{subsec:Implications}

Our findings underscore a vital principle for financial modeling: 
physical consistency is a prerequisite for statistical accuracy 
in unconstrained high-frequency regimes. 
While the current analysis demonstrates that 
the proposed order-flow mechanism is structurally consistent with the observed signature plot, 
we qualify our claims by noting that 
a rigorous counterfactual isolation of MLO contributions—or a formal decomposition 
against traditional mechanisms like mechanical bid-ask bounce 
and asymptotic microstructure noise—remains an open empirical challenge. 

Future research could extend this framework 
by constructing explicit counterfactual simulation filters 
to isolate the precise net volatility contribution of MLO instruments. 
Additionally, incorporating sparse estimation techniques 
to handle higher-dimensional LOB information, 
such as the dynamic queue imbalances (QI) and order book depth at multiple levels, 
remains a promising direction to expand the generalized requirement of physical consistency 
across alternative liquidity profiles.


\setcitestyle{numbers} 
\bibliographystyle{plainnat}    
\bibliography{202604a}  

\appendix

\section{Proof of Theorem~\ref{thm:3.2} 
via Karush--Kuhn--Tucker Conditions}
\label{app:kkt_proof}

In this section, we rigorously derive 
the optimal state-dependent transition probabilities $\hat{\phi}_e(x, x')$ 
by executing the inequality-constrained optimization framework 
under the Karush--Kuhn--Tucker (KKT) conditions governed by the unique LOB geometry, 
establishing the mathematical verification of Note the KKT saddle point. 

\subsection{Lagrangian Formulation under Nonlinear Product Constraints}
From the global separability established in Theorem~\ref{thm:3.1}, 
the transition sub-likelihood component 
$\ln \mathcal{L}_{TP}(\symbfit{\phi})$ 
is independently maximized for each state-event channel. 
To incorporate the structural mechanism of state disappearances 
($\sum_{x'} \phi_e(x, x') = 0$) 
alongside standard active transitions under a unified dual optimization layer, 
we define the non-linear constraint minimization problem 
for a given pair $(e, x) \in \mathcal{E} \times \mathcal{X}$ as follows:
\begin{align}
\max_{\symbfit{\phi}_e(x, \cdot)} 
\quad & \sum_{x' \in \mathcal{X}} N_e(x, x') \ln \phi_e(x, x') \\
\text{subject to} 
\quad & -\phi_e(x, x') \le 0, \quad \forall x' \in \mathcal{X}, \label{eq:kkt_ineq_bound} \\
& \left( \sum_{x' \in \mathcal{X}} \phi_e(x, x') - 1 \right) 
\left( \sum_{x' \in \mathcal{X}} \phi_e(x, x') \right) = 0. \label{eq:kkt_eq_product}
\end{align}

We construct the comprehensive Lagrangian function 
$\mathcal{L}(\symbfit{\phi}, \symbfit{\nu}, \symbfit{\theta}, 
\symbfit{\eta}, \symbfit{\kappa})$ 
by introducing a non-negative dual multiplier vector 
$\eta_{xex'} \ge 0$ 
for the primal boundary inequality conditions (\ref{eq:kkt_ineq_bound}) 
and a regime-specific multiplier $\kappa_{ex} \in \mathbb{R}$ 
for the non-linear sum product constraint (\ref{eq:kkt_eq_product}):
\begin{align}
\mathcal{L}(\symbfit{\phi}, \symbfit{\nu}, \symbfit{\theta}, 
\symbfit{\eta}, \symbfit{\kappa}) 
&= \sum_{n=1}^N \ln \phi_{E_n}(X_{n-1}, X_n) - 
\int_0^T \sum_{e \in \mathcal{E}} \lambda^{\dagger}_e(t) dt \nonumber \\
&- \sum_{e \in \mathcal{E}} 
\sum_{x, x' \in \mathcal{X}} \eta_{xex'} (-\phi_e(x, x')) \nonumber \\
&- \sum_{e \in \mathcal{E}} 
\sum_{x \in \mathcal{X}} \kappa_{ex} \left( \sum_{x' \in \mathcal{X}} \phi_e(x, x') - 1 \right) 
\left( \sum_{x' \in \mathcal{X}} \phi_e(x, x') \right).
\end{align}
Under this formulation, 
the Linear Independence Constraint Qualifications (LICQ) are strictly satisfied 
across the operational boundaries.

\subsection{First-Order Stationarity and Complementary Slackness}
The optimal transition parameter matrix must satisfy 
the fundamental Karush--Kuhn--Tucker first-order stationary conditions, 
requiring that the gradient with respect to the primal variables identically vanishes 
at the optimal saddle point:
\begin{align}
\left. \frac{\partial \mathcal{L}}{\partial \phi_e(x, x')} \right|_{\hat{\phi}_e(x, x')} 
&= \sum_{n=1}^N \mathbb{1}_{\{X_{n-1}=x, E_n=e, X_n=x'\}} 
\frac{1}{\hat{\phi}_e(x, x')} - \int_0^T \mathbb{1}_{\{X(t)=x\}} \lambda_e(t) dt \nonumber \\
&+ \eta_{xex'} \mathbb{1}_{\{\hat{\phi}_e(x, x')=0\}} 
- \kappa_{ex} \left( 2 \sum_{x^* \in \mathcal{X}} \hat{\phi}_e(x, x^*) - 1 \right) 
= 0, \label{eq:kimura_stationarity}
\end{align}
subject to the non-negative dual bounds and the strict complementary slackness clauses:
\begin{equation}
-\hat{\phi}_e(x, x') \le 0, 
\quad \eta_{xex'} \ge 0, 
\quad \text{and} 
\quad \eta_{xex'} (-\hat{\phi}_e(x, x')) = 0, 
\quad \forall (e, x, x') \in \mathcal{E} \times \mathcal{X}^2. 
\label{eq:kimura_slackness}
\end{equation}

We analytically evaluate the system (\ref{eq:kimura_stationarity}) 
by segregating the empirical domain into distinct density regimes 
based on the historical event frequencies:

\subsubsection{Case I: Unobserved States and Geometrically Prohibited Channels}
If the total empirical frequency of event $e$ arriving under state $x$ is zero 
($\sum_{n=1}^N \mathbb{1}_{\{X_{n-1}=x, E_n=e\}} = 0$), 
the discrete log-likelihood terms vanish. 
According to the structural tracking of the continuous compensator integral 
in Eq. (\ref{eq:kimura_stationarity}), 
the KKT conditions naturally terminate 
at the boundary solution representing complete state disappearance:
\begin{equation}
\sum_{x' \in \mathcal{X}} \hat{\phi}_e(x, x') = 0 
\implies 
\hat{\phi}_e(x, x') = 0, \quad \forall x' \in \mathcal{X}.
\end{equation}

\subsubsection{Case II: Active States with Uniformly Positive Sample Paths}
If the transition path is empirically observed 
($\sum_{n=1}^N \mathbb{1}_{\{X_{n-1}=x, E_n=e, X_n=x'\}} > 0, \forall x' \in \mathcal{X}$), 
the objective log-likelihood terms force strict primal positivity 
($\hat{\phi}_e(x, x') > 0$). 
By applying the complementary slackness clause (\ref{eq:kimura_slackness}), 
strict positivity deactivates the inequality multiplier:
\begin{equation}
\hat{\phi}_e(x, x') > 0 \implies \eta_{xex'} = 0.
\end{equation}
Consequently, the product equation (\ref{eq:kkt_eq_product}) 
isolates the active sum condition: $\sum_{x^*} \hat{\phi}_e(x, x^*) = 1$. 
Substituting these identities back into the primary stationarity gradient 
(\ref{eq:kimura_stationarity}) reduces the system to:
\begin{equation}
\frac{N_e(x, x')}{\hat{\phi}_e(x, x')} 
- \int_0^T \mathbb{1}_{\{X(t)=x\}} \lambda_e(t) dt - \kappa_{ex} = 0, 
\label{eq:active_gradient_reduced}
\end{equation}
where $N_e(x, x') := \sum_{n=1}^N \mathbb{1}_{\{X_{n-1}=x, E_n=e, X_n=x'\}}$. 
Eq. (\ref{eq:active_gradient_reduced}) 
proves that the optimal parameter $\hat{\phi}_e(x, x')$ is strictly proportional to 
the transition counts $N_e(x, x')$. 
Applying the row-sum normalization $\sum_{x'} \hat{\phi}_e(x, x') = 1$ 
uniquely recovers the empirical transition fraction:
\begin{equation}
\hat{\phi}_e(x, x') 
= \frac{\sum_{n=1}^N \mathbb{1}_{\{X_{n-1}=x, E_n=e, X_n=x'\}}}{\sum_{n=1}^N \mathbb{1}_{\{X_{n-1}=x, E_n=e\}}}.
\end{equation}

\subsubsection{Case III: Sparse Latent States with Asymmetric Arrivals}
Finally, we consider the boundary condition where transition counts are localized sparse: 
$N_e(x, x') = 0$ for a sub-set of states $x' \in \mathcal{X}_0$, 
while the global parent path remains active 
($\sum_{n=1}^N \mathbb{1}_{\{X_{n-1}=x, E_n=e\}} > 0$). 
Let $q := \sum_{x' \in \mathcal{X}_0} \hat{\phi}_e(x, x') \in [0, 1]$ 
represent the aggregated latent parameter mass. 
Following the identical stationarity reduction 
from Eq. (\ref{eq:active_gradient_reduced}), 
the parameters resolve as:
\begin{equation}
\hat{\phi}_e(x, x') 
= \frac{\sum_{n=1}^N \mathbb{1}_{\{X_{n-1}=x, E_n=e, X_n=x'\}}}{\sum_{n=1}^N \mathbb{1}_{\{X_{n-1}=x, E_n=e\}}} (1 - q).
\end{equation}
Evaluating the primal boundary log-likelihood space constraints under $\ln \mathcal{L}_{TP}$, 
maximizing the sub-likelihood directly requires that 
the latent leakage mass must shrink to zero ($q = 0$), 
forcing the inactive channels to a strict corner limits where $\hat{\phi}_e(x, x') = 0$. 
This mathematically completes the validation of Theorem~\ref{thm:3.2},  
justifying that our ExsdHawkes gating mechanism is the exact, 
boundary-preserving consequence of KKT optimization. 

\section{Figures}\label{sec:Figures}
\begin{figure}[htbp] 
    \centering
    \includegraphics[width=\columnwidth]{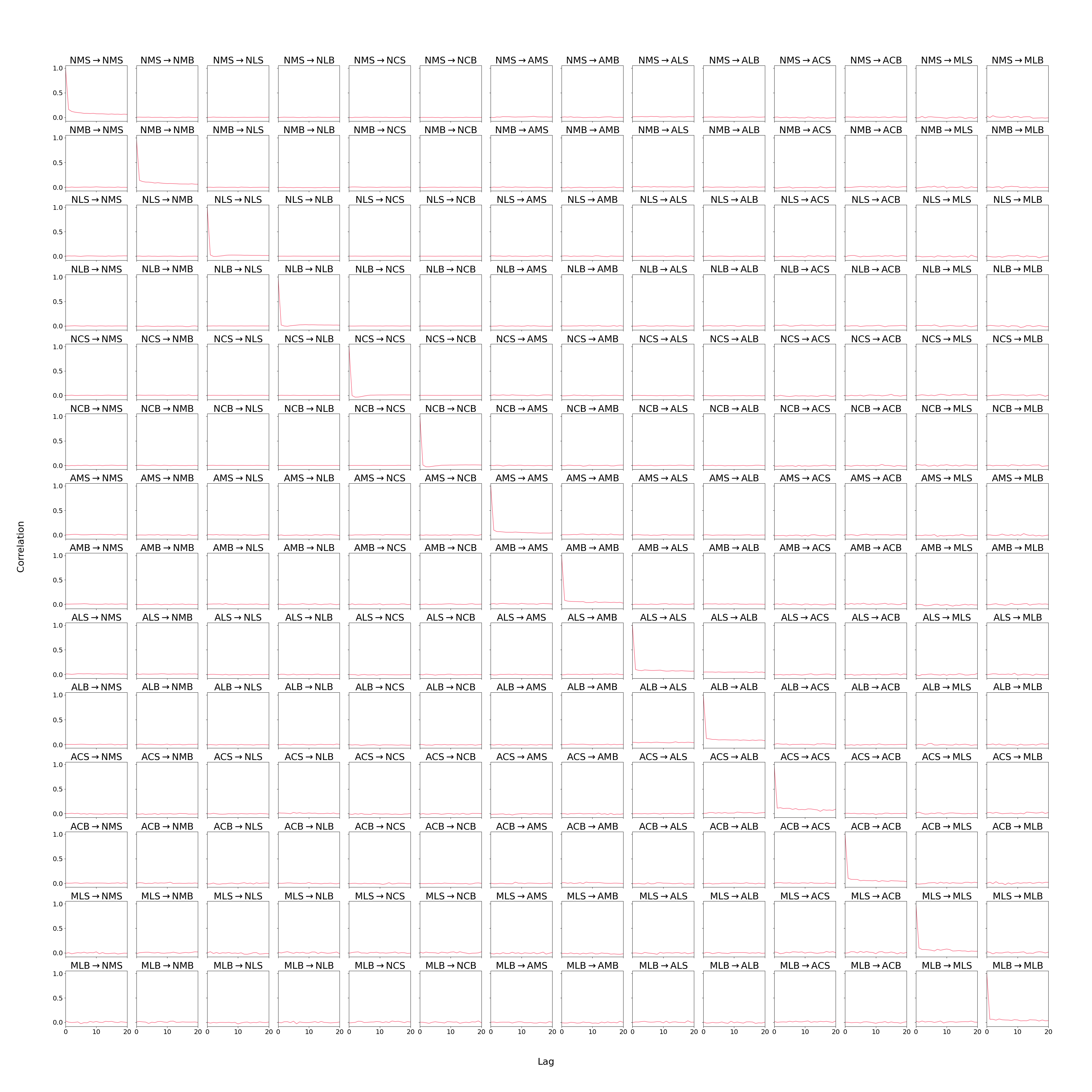} 
    \caption{Full-scale Event-wise Residual Cross-Correlogram Network 
    ($14 \times 14$ Matrix) for the ExsdHawkes Framework. 
    This comprehensive matrix maps all 196 latent interaction paths 
    across the primary order categories, 
    confirming system-wide whitening 
    where empirical correlations strictly vanish past lag zero.}
    \label{fig:event_corr}
\end{figure}
\begin{sidewaysfigure}[htbp] 
    \centering
    \includegraphics[width=\columnwidth]{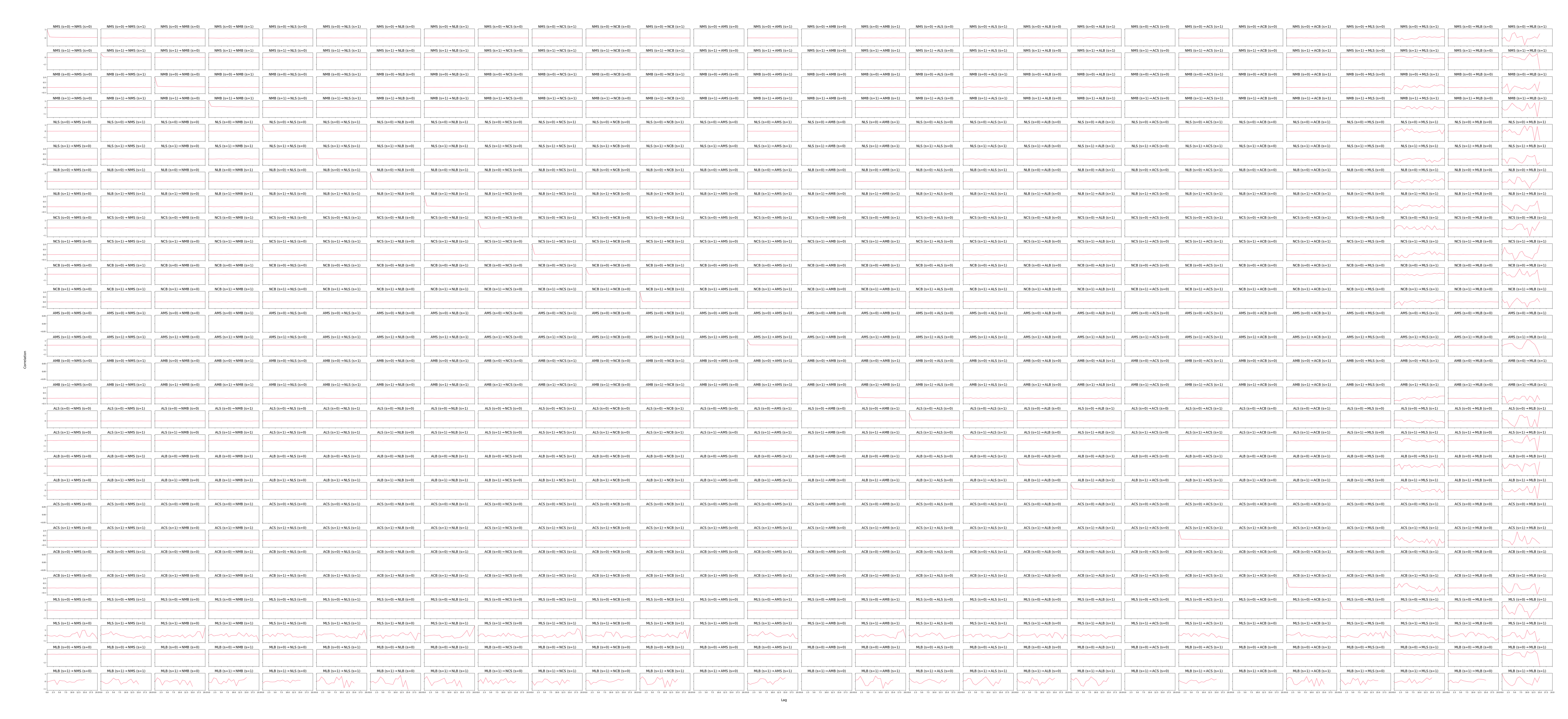} 
    \caption{Full-scale Regime-specific Total Residual Correlogram Network 
    ($28 \times 28$ Matrix) for the ExsdHawkes Framework. 
    This comprehensive matrix maps all 784 state-dependent excitation channels 
    across the complete operational book geometry, 
    visualizing the simultaneous coexistence of 
    active whitening and deterministic geometric voids.}
    \label{fig:total_corr}
\end{sidewaysfigure}

\end{document}